\documentclass[journal,final]{IEEEtran}

\usepackage{etoolbox}
\usepackage{graphicx}
\usepackage{amssymb}
\usepackage{amsmath}
\usepackage{multicol}
\usepackage{booktabs}
\usepackage{subcaption}
\usepackage{cite}
\usepackage{url}
\usepackage[utf8]{inputenc}
\usepackage{amsthm}
\usepackage{enumitem}
\usepackage[table]{xcolor}
\usepackage{color}             
\usepackage{mathtools}         
\newtheorem{lemma}{Lemma}

\newtheorem{theorem}{Theorem}
\newtheorem{definition}{Definition}

\newcommand{\mypar}[1]{{\bf #1.}}
\newcommand{\coord}[1]{\text{\bf #1}}
\newcommand{\C}[0]{{\mathbb{C}}}
\newcommand{\R}[0]{{\mathbb{R}}}

\newcommand{\DLT}[0]{\operatorname{DLT}}

\newcommand{\diag}[0]{\operatorname{diag}}

\newcommand{\ra}[1]{\renewcommand{\arraystretch}{#1}}

\newcommand{\meet}[0]{\wedge}
\newcommand{\bigmeet}[0]{\bigwedge}
\newcommand{\join}[0]{\vee}

\newcommand{\ft}[1]{\widehat{#1}}
\newcommand{\fr}[1]{\overline{#1}}
\newcommand{\chr}[1]{\iota_{\{ #1 \} }}
\newcommand{\TV}[0]{\operatorname{TV}}
\newcommand{\STV}[0]{\operatorname{STV}}
\newcommand{\latt}[0]{{\cal L}}
\newcommand{\sisu}[0]{{\cal A}}
\newcommand{\ftsu}[0]{{\cal B}}
\newcommand{\edges}[0]{{\cal E}}
\newcommand{\gens}[0]{{\cal G}}
\newcommand{\TVs}[0]{{\cal T}}
\newcommand{\minel}[0]{\text{\it min}}
\newcommand{\maxel}[0]{\text{\it max}}

\newcommand{\obj}[0]{\operatorname{obj}}
\newcommand{\attr}[0]{\operatorname{attr}}

\newcommand{\I}              {\mathrm{j}}

\DeclarePairedDelimiter\ideal\langle\rangle
\reDeclarePairedDelimiterInnerWrapper\ideal{star}{
\mathopen{#1\vphantom{\MTkillspecial{#2}}\kern-\nulldelimiterspace\right.}
#2
\mathclose{\left.\kern-\nulldelimiterspace\vphantom{\MTkillspecial{#2}}#3}}

\DeclarePairedDelimiter{\abs}{\lvert}{\rvert}

\DeclarePairedDelimiter{\norm}{\lVert}{\rVert}

\begin{document}

\title{Discrete Signal Processing \\on Meet/Join Lattices}
%
%
%
\author{Markus P{\"u}schel,~\IEEEmembership{Fellow,~IEEE}, Bastian Seifert,~\IEEEmembership{Member,~IEEE}, and Chris Wendler,~\IEEEmembership{Student Member,~IEEE}%
\thanks{The authors are with the Department of Computer Science, ETH
  Zurich, Switzerland (email: pueschel@inf.ethz.ch,
  bastian.seifert@inf.ethz.ch, chris.wendler@inf.ethz.ch)}}
\maketitle

\begin{abstract}
A lattice is a partially ordered set supporting a meet (join) operation that returns the largest lower bound (smallest upper bound) of two elements. Just like graphs, lattices are a fundamental structure that occurs across domains including social data analysis, natural language processing, computational chemistry and biology, and database theory. In this paper we introduce discrete-lattice signal processing (DLSP), an SP framework for data, or signals, indexed by such lattices. We use the meet (or join) to define a shift operation and derive associated notions of filtering, Fourier basis and transform, and frequency response. We show that the spectrum of a lattice signal inherits the lattice structure of the signal domain and derive a sampling theorem. Finally, we show two prototypical applications: spectral analysis of formal concept lattices in social science and sampling and Wiener filtering on multiset lattices in combinatorial auctions. Formal concept lattices are a representation of relations between objects and attributes. Since relations are equivalent to bipartite graphs and hypergraphs, DLSP offers a form of Fourier analysis for these structures.
\end{abstract}


%
\IEEEpeerreviewmaketitle


\section{Introduction}

\IEEEPARstart{W}{e} have entered the age of big data and signal processing (SP) as a classical data science is a key player for data analysis. One main challenge is to broaden the foundation of SP to become applicable across the large variety of data available.

Classical SP and many of its applications are built on the fundamental concepts of time- or translation-invariant systems, convolution, Fourier analysis, sampling, and others. But many datasets are not indexed by time, or its separable extensions to 2D or 3D. Thus an avenue towards broadening SP is in porting these concepts to index domains with inherently different structures to enable the power of the numerous SP techniques that build on these concepts.

\mypar{Graph signal processing and beyond} A prominent example of this approach is graph signal processing (GSP) \cite{Sandryhaila:13,Shuman:13}, which considers data indexed by the nodes of a graph. It replaces the time shift or translation by the adjacency or Laplacian operator, whose eigendecomposition yields the Fourier transform. Polynomials in the operator yield the associated notion of convolution. Building on these concepts, numerous follow-up works have reinterpreted and ported many classical SP tools to GSP \cite{Ortega.Frossard.Kovacevic.Moura.Vandergheynst:2018a}. In machine learning, graph convolutions have yielded new types of neural networks \cite{Bronstein:17}.

However, not all index domains are graphs or a graph may be the wrong abstraction. Thus, there have been various recent works that aim to further broaden SP beyond graphs. Examples include SP on hypergraphs \cite{Zhang:20}, graphons \cite{Ruiz:20}, simplicial complexes \cite{Barbarossa:20}, so-called quivers \cite{Paradamayorga:20}, and powersets \cite{Pueschel:20}.
Some of the above works have been inspired by the algebraic signal processing theory (ASP)~\cite{Pueschel:08a,Pueschel:06c}, which provides the axioms and a general approach to obtaining new SP frameworks. In particular, ASP identifies the shift operator (or operators) as key axiomatic concept from which convolution, Fourier analysis, and other concepts can be derived. Early examples derived SP frameworks on various path graphs~\cite{Pueschel:08b,Sandryhaila:12} or regular grids \cite{Pueschel:07}.

\mypar{Meet/join lattices} In this paper we consider signals, or data, indexed by meet/join lattices, a fundamental structure across disciplines. A lattice is a partially ordered set that supports a meet (join) operation that returns for each pair of elements their largest lower bound (smallest upper bound). We provide a few examples.

Formal concept lattices~\cite{Ganter.Wille:2012a} are used in social data analysis. Such a lattice distills relations between objects (e.g., users) and attributes (e.g., likes a certain website) into so-called concepts. Relations can be viewed as tables with binary entries and are equivalent to bipartite graphs or hypergraphs. Formal concept lattices also play a role in the emerging field of granular computing in information processing~\cite{Wei:16}.

Ranked data~\cite{Monjardet.Raderanirina:2004a} can be viewed as indexed by the set of permutations, which form a lattice, partially ordered by the number of neighbor exchanges they consist of.

The set of partitions of a finite set form a lattice, which is fundamental in database theory \cite{Lee:83}.

In combinatorial chemistry some reaction networks form a
lattice~\cite{Klein.Ivanciuc.Ryzhov.Ivanciuc:2008a}, whose elements are chemical structures. Lattice signals can then represent chemical properties, e.g., retention indices.

In evolutionary genetics a gene and its genotypes, i.e., variations by mutations, form a lattice \cite{Alberch:91, Beerenwinkel:11}. Phenotypes, such as the resistance to certain drugs, are then naturally represented by lattice signals \cite{Bennett:09}.


In natural language processing lattices occur naturally, e.g., the set of words with the prefix order and a meet operation that returns the longest common prefix~\cite{Birget:93}.


%

\mypar{Contribution: Discrete-lattice SP (DLSP)} In this paper we present a novel linear SP framework, called DLSP, for signals indexed by finite meet/join lattices. First, we derive the theory using the general high-level procedure provided by ASP \cite{Pueschel:08a}. Namely, we use the meet (or join) operation to define a notion of shift, from which we derive the associated notions of convolution, Fourier basis, Fourier transform, frequency response, and others. Using the concept of total variation, we show that the spectrum is again partially ordered and forms a lattice isomorphic to the signal domain. We derive a sampling theorem for signals with known sparse Fourier support.

In the second part we show two possible application domains: spectral analysis of formal concepts lattices occurring in social data analysis and sampling and Wiener filtering on multiset lattices occurring in combinatorial auctions. Formal concept lattices are particularly interesting as they distill relations between objects and attributes into a lattice. Since relations are equivalent to bipartite graphs and hypergraphs, DLSP offers a form of Fourier analysis for data indexed by these structures.

This paper extends and completes the preliminary work in \cite{Pueschel:19,Wendler:19,Seifert.Wendler.Pueschel:2021a}.

\section{Background on Lattices}

In this section we provide the necessary basic background on semilattices and lattices. Lattice theory is a well-studied area in mathematics. Good reference books are \cite{Stanley:11,Graetzer:11}, which also give an idea of the rich theory available for lattices.

\mypar{Partially ordered sets}
We consider finite sets $\latt$ with a {\em partial order} $\leq$, also called {\em posets}. We denote elements of $\latt$ with lower case letters $a,b,x,\dots$. Formally, a poset satisfies three axioms for all $a,b,c\in\latt$: 
\begin{enumerate}
	\item $a\leq a$,
	\item $a\leq b$ and $b\leq a$ implies $a=b$, and
	\item $a\leq b$ and $b\leq c$ implies $a\leq c$.
\end{enumerate}
We will use $a<b$ if $a\leq b$ but $a\neq b$.

We say that $b$ covers $a$, if $a < b$ and there is no $x\in L$ with $a<x<b$. Equivalently, $a$ is a largest element (there could be several) below $b$. 

\mypar{Meet-semilattice}
A meet-semilattice is a poset $\latt$ that also permits a meet operation $a\meet b$, that returns the greatest lower bound of $a$ and $b$, which must be unique. Formally, it satisfies for all $a,b,c\in\latt$:
\begin{enumerate}
\item $a\meet a = a$,
\item $a\meet b = b\meet a$, and 
\item $(a\meet b)\meet c = a\meet(b\meet c)$. 
\end{enumerate}
The latter two show that $\meet$ is commutative and associative.

Using the notion of cover, semilattices can be visualized through a directed acyclic graph $(\latt, \edges)$ called cover graph. The nodes are the elements of $\latt$ and $(b,a)\in\edges$ if $b$ covers $a$. The elements of $\edges$ are called covering pairs. The graph is typically drawn in a way that if $a\leq b$, then $b$ is higher than $a$. 

\begin{figure*}\centering
	\subcaptionbox{Some meet-semilattice\label{l1}}[0.24\linewidth]
	{\includegraphics[scale=0.5]{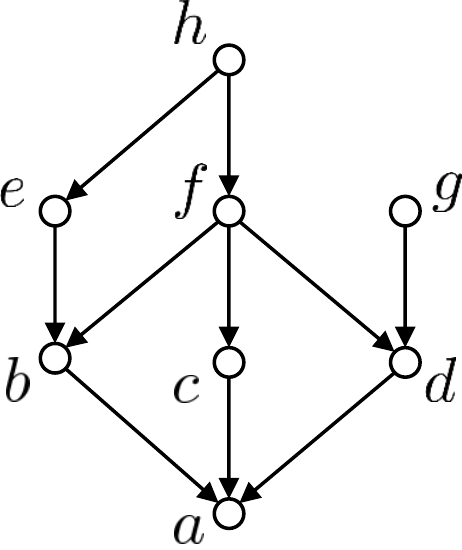}}
	\subcaptionbox{Total order lattice\label{l2}}[0.24\linewidth]
	{\includegraphics[scale=0.5]{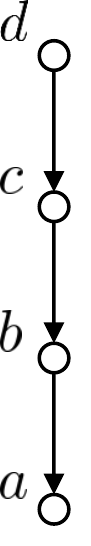}}
	\subcaptionbox{Subset lattice\label{l3}}[0.24\linewidth]
	{\includegraphics[scale=0.5]{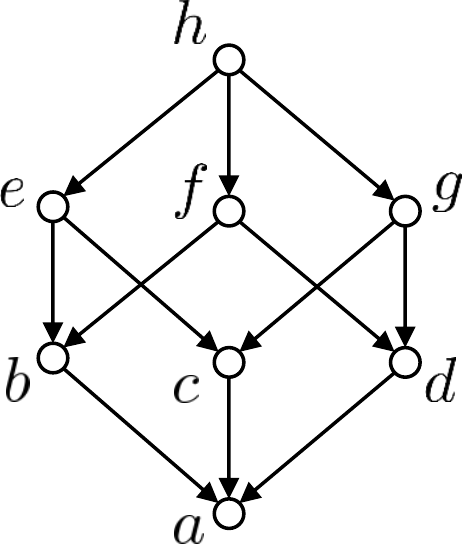}}
	\subcaptionbox{Not a semilattice\label{nl}}[0.24\linewidth]
	{\includegraphics[scale=0.5]{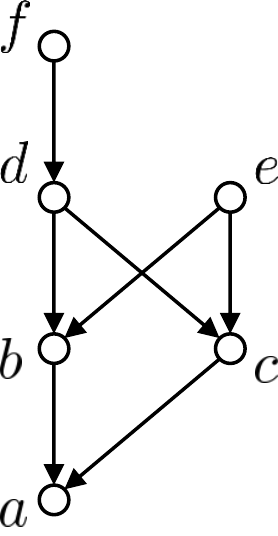}}
	\caption{Examples of semilattices $\latt$. The labels are the names of the nodes, i.e., elements of $\latt$. In this paper, we consider signals on $\latt$ that associate values which each node.\label{semilattices}}
\end{figure*}

Fig.~\ref{l1} shows a small example of a meet-semilattice. The cover graph shows, for example, that $c\leq f, b\leq h, c\not\leq g$ and $e\meet f = b, h\meet g = a$. Note that a meet semilattice must have a unique smallest element $\minel$, which is the meet of all elements in $\latt$. In Fig.~\ref{l1} the smallest element is $\minel = a$.

\mypar{Lattice}
Dual to the meet-semilattice, one can define a {\em join-semilattice} as a poset with an operation $\join$ that returns the least upper bound and satisfies analogous properties and necessarily has a unique maximal element. If a finite poset supports both a meet and a join operation, it is called a {\em lattice}. Fig.~\ref{l1} is not a join semilattice since the join of $g$ and $h$ does not exist. A join semilattice has a unique maximum (the join of all elements).

In the case of finite lattices $\latt$ considered here, the difference between semilattices and lattices is marginal. More precisely, the following holds:
\begin{lemma}\label{addmax}
Every meet-semilattice with a unique maximal element $\maxel$, i.e., $\maxel\geq a$ for all $a\in\latt$, is a lattice. Namely, if $a^u = \{x\in\latt\mid x\geq a\}$, then the join of $a$ and $b$ can be computed as
$$
a\join b = \bigmeet a^u\cap b^u.
$$
\end{lemma}
The maximum ensures that $a^u\cap b^u$ is not empty. As a consequence we can always extend a meet-semilattice to a lattice by adding only one artificial maximal element if it does not exist already.

\mypar{Examples} 
If the partial order is a total order, i.e., for all distinct $a,b\in\latt$ either $a\leq b$ or $b\leq a$, then $\latt$ is a lattice. An example is shown in Fig.~\ref{l2} with four elements.

Another special case of a lattice is the powerset (set of all subsets) of a finite set ${\cal S}$ with $\cap$ as meet and $\cup$ as join, which yields a directed hypercube as associated cover graph. Fig.~\ref{l3} shows the example for ${\cal S} = \{\alpha_1,\alpha_2,\alpha_3\}$: here, $a = \emptyset$, $b,c,d$ are the one-element subsets, $e,f,g$ the two-element subsets, and $h = {\cal S}$. Signals indexed by powersets are called set functions. We introduced an SP framework for set functions in \cite{Pueschel:18,Pueschel:20}. The work in this paper generalizes to arbitrary meet- or join-semilattices.

Not every poset with a unique minimal element is a meet-semilattice. For example, the poset in Fig.~\ref{nl} is not since the meet of $d$ and $e$ is not unique.

We will provide more concrete examples of lattices, including multiset and concept lattices, with applications in Sections~\ref{sec:FormalConceptLattices} and \ref{sec:MultisetAuctions}.

\mypar{Some properties}
The cover graph associated with a meet-semilattice is acyclic and, as said above, has a unique sink (node with outdegree 0). However, these properties are not sufficient to make a directed graph (digraph) a lattice, as shown in Fig.~\ref{nl}.

Every poset can be topologically sorted, i.e., it can be totally ordered in a way that is compatible with the partial order. In Fig.~\ref{semilattices} this is done with the lexicographic order of the node labels, meaning smaller elements in the poset come earlier in the alphabet. In general, the topological sort is not unique.

\section{Discrete-Lattice SP}

In this section we introduce discrete-lattice SP\footnote{More correct would be discrete-semilattice SP, but we opted for simplicity due to Lemma~\ref{addmax}.} (DLSP), a signal processing framework for signals, or data, associated with the elements of a given meet- or join-semilattice. DLSP could also be viewed as a form of graph SP on the special type of directed acyclic graphs that are associated with semilattices (e.g., those in Fig.~\ref{semilattices}(a)--(c)), but is based on a very different notion of shift. We discuss the differences to prior digraph SP frameworks later in Section~\ref{sec:discussion}.

Our derivation of DLSP follows the general procedure of ASP \cite{Pueschel:08a,Pueschel:06c} for deriving linear SP frameworks. We start with the shift definition, from which we then derive the associated notion of convolution or filtering, Fourier transform, frequency response, and a convolution theorem. 

In Section~\ref{sec:example} we then provide a detailed example.

\subsection{Lattice Signals and Filters}

\mypar{Lattice signals} We consider real signals indexed with the elements of a given meet-semilattice $\latt$ of size $n$:
\begin{equation}
\coord{s}:\ \latt\rightarrow\R,\quad x\mapsto s_x.
\end{equation}
We will also write $\coord{s} = (s_x)_{x\in\latt}$, considered as a column vector, where we assume a specific order in which the elements are topologically sorted from small to large. The set of signals is an $n$-dimensional vector space.

\mypar{Shifts} Our construction uses the meet operation to define a shift. Formally, for every $a\in \latt$ we define the
\begin{equation}\label{shiftdef}
\text{shift by $a$:}\quad (s_{x})_{x\in\latt}\rightarrow (s_{x\meet a})_{x\in\latt},
\end{equation}
which captures the semilattice structure. Obviously, this shift is a linear mapping on the signal space since shifting $\alpha\coord{s}+\beta\coord{s}'$ ($\alpha,\beta\in\R$) by $a$ yields $\alpha(s_{x\meet a})_{x\in\latt}+\beta(s'_{x\meet a})_{x\in\latt}$. Thus, \eqref{shiftdef} can be represented by a matrix $T_a$ such that
\begin{equation}\label{shiftdefmat}
T_a\coord{s} = (s_{x\meet a})_{x\in\latt}.
\end{equation}
Since $\meet$ is commutative, the definition implies that all shifts commute:
\begin{equation}\label{shiftprop}
T_a(T_b\coord{s}) = T_{a\meet b}\coord{s} = T_b(T_a\coord{s}),\quad\text{for }a,b\in\latt.
\end{equation}
If $\latt$ contains a maximal element $\maxel$, then $x\meet\maxel = x$ for all $x\in\latt$ and thus
\begin{equation}\label{maxshift}
T_\maxel\coord{s} = \coord{s}
\end{equation}
is the identity mapping.

\mypar{Generating shifts}
In discrete-time SP a signal can be shifted by any integer value, but each such shift can be expressed as a repeated shift by 1, i.e., the shift by 1 {\em generates} all others. The question is which shifts in our lattice are needed to generate all others. Equation \eqref{shiftprop} shows that a shift by $a\meet b$ can be done by shifting by $a$ and then by $b$ and thus is not needed as generator. Thus, as explained by lattice theory \cite{Graetzer:11}, the generators are precisely all meet-irreducible elements, i.e., those $a\in\latt$ that cannot be written as $a = b\meet c$ with $b,c\neq a$. If $\latt$ has a (necessarily) unique maximal element $\maxel$ it is not considered as meet-irreducible, since shifting by $\maxel$ is the identity (see \eqref{maxshift}).

In Fig.~\ref{l1} the basic shifts that generate all others are given by $h, e, f, g, c$. In Fig.~\ref{l2} these are all elements except $d$, and in Fig.~\ref{l3} they are $e, f, g$.

\mypar{Filters and convolution} The associated notion of convolution is obtained by linearly extending the shift operation to filters $\sum_{a\in\latt}h_aT_a$. Convolution then is defined as
$$
\coord{h}\ast\coord{s} = 
\left(\sum_{a\in\latt}h_aT_a\right)\coord{s} = \left(\sum_{a\in\latt}h_as_{x\meet a}\right)_{x\in\latt}.
$$

\mypar{Shift-invariance} Since by \eqref{shiftprop} all shifts commute, each shift also commutes with all filters. So filters in DLSP are linear shift-invariant mappings.

\subsection{Spectrum and Fourier Transform}	

\mypar{Fourier basis and frequency response} To derive the Fourier transform on lattices we first determine the Fourier basis, i.e., those signals that are simultaneous eigenvectors for all shifts (in particular the generating shifts) and thus for all filters. Note that their existence is possible since by \eqref{shiftprop} the shifts commute. Lattice theory confirms the existence of a joint eigenbasis through the zeta transform \cite{Rota:64}. Translated to our SP setting, we obtain a complete Fourier basis with $n$ vectors, each associated with a lattice element. 

\begin{theorem}[Fourier basis]\label{fourbasis}
Let $\chr{y\leq x}$ denote the characteristic (or indicator) function of $y\leq x$: 
$$
\chr{y\leq x} = \chr{y\leq x}(x,y) = \begin{cases}1, &  y\leq x,\\ 0, & \text{else}.\end{cases}
$$
For each $y\in\latt$, the vector
\begin{equation}\label{purefreq}
\coord{f}^y = (\chr{y\leq x})_{x\in\latt}
\end{equation}
is an eigenvector of all shift matrices $T_a$, $a\in\latt$. Specifically,
\begin{equation}\label{evs}
	T_a\coord{f}^y = (\chr{y\leq x\meet a})_{x\in\latt} = 
	\begin{cases}1\cdot\coord{f}^y=\coord{f}^y, & y\leq a,\\ 0\cdot\coord{f}^y=\coord{0}, & y\not\leq a.\end{cases}
\end{equation}
Further, the $\coord{f}^y$ are linearly independent and thus form a basis.
\end{theorem}
Before we start the proof we note that $\coord{f}^y$ consists only of entries 1 and 0, depending on the condition $y\leq x$. In particular, for the unique minimal element $\minel\in\latt$,
$$
\coord{f}^\minel = (\chr{\minel\leq x})_{x\in\latt} = (1, 1, \dots ,1)^T
$$
is the constant vector.
\begin{proof}
Let $y,a\in\latt$, then
$$
T_a\coord{f}^y = (\chr{y\leq x\meet a})_{x\in\latt}.
$$
Case 1: $y\leq a$. Then $y\leq x\meet a\Leftrightarrow y\leq x$ ($\Rightarrow$ is obvious; for $\Leftarrow$, $y\leq a$ and $y\leq x$ implies $y\leq x\meet a$) and we get
$$
(\chr{y\leq x\meet a})_{x\in\latt} = (\chr{y\leq x})_{x\in\latt} = \coord{f}^y.
$$

Case 2: $y\not\leq a$. Then $y\leq x\meet a$ never holds and we get
$$
(\chr{y\leq x\meet a})_{x\in\latt} = \coord{0}.
$$

Finally, we note that the eigenmatrix $(\coord{f}^y)_{y\in\latt}$ is lower triangular assuming the columns are ordered by a topological sort of $\latt$. The diagonal entries $\chr{y\leq y}$ are all $=1$ and thus the $\coord{f}^y$ are linearly independent and form a Fourier basis.
\end{proof}
Equation~\eqref{evs} shows that the frequency response of a shift at frequency $y$ is either $1$ or $0$. By linear extension, the frequency response of a filter $\coord{h} = (h_a)_{a\in\latt}$ at frequency $y$ is computed as
\begin{equation}\label{freqresp}
\fr{h}_y = \sum_{a\in\latt,a\geq y}h_a.
\end{equation}
Note that frequency response and Fourier transform are computed differently.

\mypar{Frequency ordering}
An important aspect for a useful SP framework is the ordering of frequencies. We build on the idea in \cite{Sandryhaila:13}, which relates total variation and shift in graph SP. Here we need an extension to our multiple-shift setting. We first denote the set of meet-irreducible elements (corresponding to the generating shifts) as $\gens = \{g_1,\dots,g_k\}\subset \latt$.

\begin{definition}
	We define the total variation of a signal $\coord{s}$ w.r.t. a given shift $g\in G$ as
	\begin{equation}\label{tvdefg}
	\TV_{g}(\coord{s}) = ||\coord{s} - T_{g}\coord{s}||_2.
	\end{equation}
	Further, we define the total variation as the $k$-tuple of the individual $\TV_g$:
	\begin{equation}\label{tvdef}
	\TV(\coord{s}) = (\TV_{g}(\coord{s}))_{g\in\gens}.
	\end{equation}
	Since for applications it is convenient to only have one number, we also define the sum total variation as
	\begin{equation}\label{stvdef}
	\STV(\coord{s}) = \sum_{g\in\gens} \TV_{g}(\coord{s}).
	\end{equation}
\end{definition}

\begin{theorem}\label{tv-th}
	The total variation of the Fourier basis vector $\coord{f}^y$, $y\in\latt$, when normalized to $||\coord{f}^y||_2 = 1$, is given by the $k$-tuple
	\begin{equation}\label{tv}
	\TV(\coord{f}^y) = (\chr{y\not\leq g})_{g\in\gens} 
	\end{equation}
	and thus
	$$
	\STV(\coord{f}^y) = |\{g\in\gens\mid y\not\leq g\}|.
	$$
	In particular
	$$
	0\leq\STV(\coord{f}^y)\leq |\gens| = k.
	$$
	Further, the set of total variations $\TVs = \{\TV(\coord{f}^y)\mid y\in\latt\}$ is a poset w.r.t.~componentwise comparison and isomorphic to $\latt$: for $x,y\in\latt$,
	\begin{equation}\label{iso}
	x\leq y\text{ if and only if }\TV(\coord{f}^x)\leq \TV(\coord{f}^y),
	\end{equation}
	which means the frequencies inherit the partial order of $\latt$. In particular, the lowest frequency is $\coord{f}^\minel$ with $\TV(\coord{f}^\minel) = (0,\dots,0)$. $\STV$ yields a topological sorting of $\TV$.
\end{theorem}
\begin{proof}
$TV_g(\coord{f}^y) = ||\coord{f}^y - T_g\coord{f}^y||_2$. Using \eqref{evs}, this yields 1 if $y\not\leq g$ and 0 else, which proves \eqref{tv}.

For \eqref{iso} let $x\leq y$ and let $g\in\gens$. Then $y\leq g$ implies $x\leq g$, i.e., $x\not\leq g$ implies $y\not\leq g$, which yields $\TV_g(\coord{f}^x) = (\chr{x\not\leq g})_{g\in\gens}\leq(\chr{y\not\leq g})_{g\in\gens} = \TV_g(\coord{f}^y)$. 

For the converse we first note that every element in $y\in\latt$ can be written as the meet of $g\in\gens$, which are necessary those larger than $y$: $y = \bigmeet_{g\in\gens,y\leq g}g$. Since $\TV(\coord{f}^x)\leq\TV(\coord{f}^y)$ means that $y\leq g$ implies $x\leq g$, we obtain $x\leq y$ as desired.
\end{proof}

We note that the defined frequency ordering is, in a sense, unique, since it is not affected by possible variants in the definition of total variation. First, \eqref{tv} stays the same if the the $2$-norm in \eqref{tvdefg} is replaced by any $p$-norm, $1\leq p < \infty$, and the Fourier basis vectors are normalized w.r.t.~the same norm. Second, since the total variation of a Fourier basis vector in \eqref{tv} consists only of 0s and 1s, the total ordering of the sum total variations is not affected if the definition in \eqref{stvdef} is replaced by a squared sum or any $p$-norm, $1\leq p < \infty$.

\mypar{Fourier transform} We denote the spectrum (Fourier coefficients) of a signal $\coord{s}$ as $\ft{\coord{s}} = (\ft{s}_y)_{y\in\latt}$. Equation~\eqref{purefreq} shows that the inverse Fourier transform is given by
\begin{equation}\label{invdlt}
s_x = \sum_{y\in\latt}(\chr{y\leq x})\ft{s}_y
= \sum_{y\leq x}\ft{s}_y,\quad x\in\latt.
\end{equation}
This equation is inverted using the classical Moebius inversion formula (see \cite[Sec.~3]{Rota:64}) and yields the associated Fourier transform that we call {\em discrete lattice transform} ($\DLT$ or $\DLT_\latt$):
\begin{equation}\label{dlt}
\ft{s}_y = \sum_{x\leq y}\mu(x,y)s_x.
\end{equation}
Here, $\mu$ is the Moebius function, defined recursively as
\begin{align*}
\mu(x,x) & = 1, & \text{for }x\in\latt,\\
\mu(x,y) & = -\sum_{x\leq z < y}\mu(x,z), & x\neq y.
\end{align*}
In matrix form,
$$
\ft{\coord{s}} = \DLT_\latt\coord{s},\quad\text{with }\DLT_\latt = [\mu(x,y)\chr{x\leq y}]_{y,x\in\latt}.
$$

\begin{figure*}\centering
	\subcaptionbox{}[0.24\linewidth] 
	{
		\includegraphics[width=\linewidth]{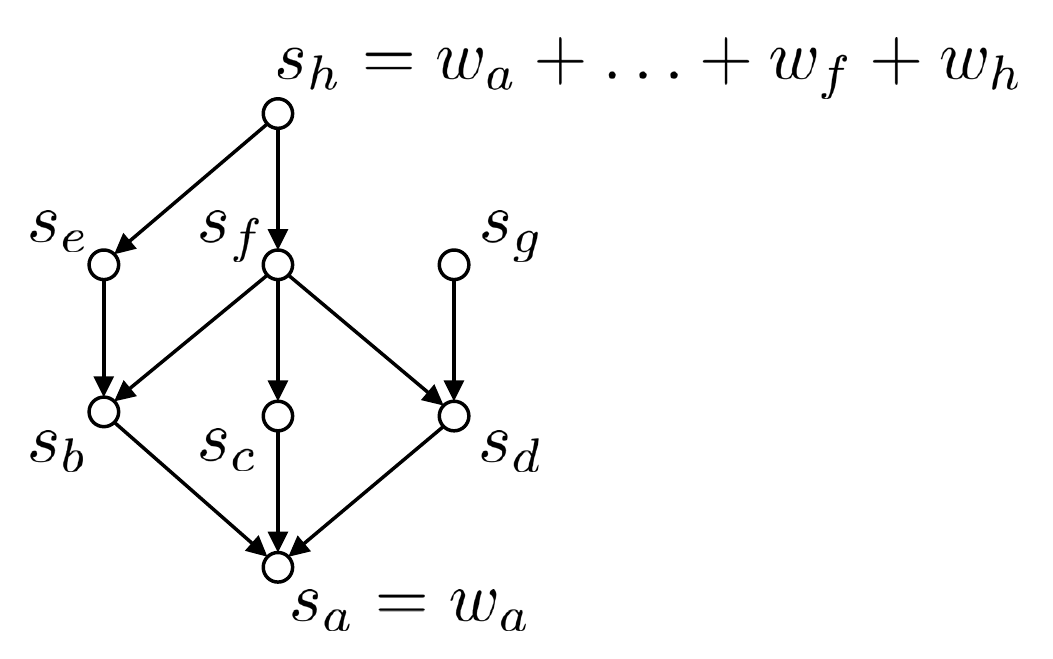}
	}
	\subcaptionbox{}[0.30\linewidth] 
	{ \resizebox{!}{14mm}{
			\begin{tabular}{@{}c@{\,}l@{}}\toprule
				\multicolumn{2}{@{}l}{Signal}\\
				\midrule 
				$s_h$ & $= w_a + w_b + w_c + w_d + w_e + w_f + w_h$ \\
				$s_g$ & $ = w_a + w_d + w_g$ \\
				$s_f$ & $ = w_a + w_b + w_c + w_d + w_f$ \\
				$s_e$ & $ = w_a + w_b + w_e$ \\
				$s_d$ & $ = w_a + w_d$ \\
				$s_c$ & $ = w_a + w_c$ \\
				$s_b$ & $ = w_a + w_b$ \\
				$s_a$ & $ = w_a$ \\
				\bottomrule
			\end{tabular}
		}
	}
	\subcaptionbox{}[0.22\linewidth] 
	{
		\includegraphics[width=\linewidth]{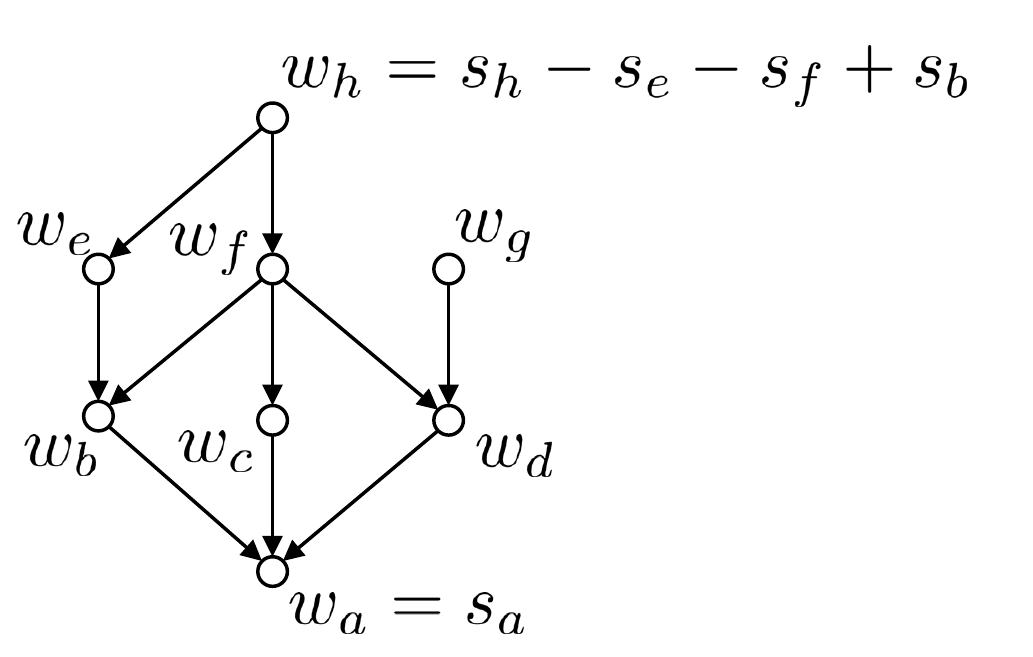}
	}
	\subcaptionbox{}[0.20\linewidth] 
	{ \resizebox{!}{14mm}{
			\begin{tabular}{@{}c@{\,}l@{}}\toprule
				\multicolumn{2}{@{}l}{Spectrum}\\
				\midrule 
				$w_h$ & $ = s_h - s_f - s_e + s_b$ \\
				$w_g$ & $ = s_g - s_d$ \\
				$w_f$ & $ = s_f - s_d - s_c - s_b + 2 s_a$ \\
				$w_e$ & $ = s_e - s_b$ \\
				$w_d$ & $ = s_d - s_a$ \\
				$w_c$ & $ = s_c - s_a$ \\
				$w_b$ & $ = s_b - s_a$ \\
				$w_a$ & $ = s_a$\\
				\bottomrule
			\end{tabular}
		}
	}
	\caption{Intuition for spectrum and Fourier transform. The inverse DLT computes the signal values in (a) by summing the weights of all predecessors (b). The DLT recovers the weights (c) from the signal values (d).}
	\label{fig:IntuitionFourierSpectrum}
\end{figure*}

\begin{table*}\centering
	\ra{1.2}
	\begin{tabular}{@{}lll@{}}\toprule
		Concept & Meet-semilattice & Join-semilattice\\ \midrule
		Shift by $q\in\latt$ & $(s_{x\meet q})_{x\in\latt}$ & $(s_{x\join q})_{x\in\latt}$ \\
		Generating shifts & meet-irr. $q\in\latt$ & join-irr. $q\in\latt$ \\
		Convolution & $\sum_{q\in\latt}h_qs_{x\meet q}$ & $\sum_{q\in\latt}h_qs_{x\join q}$ \\
		DLT $\ft{s}_y =$ & $\sum_{x\leq y}\mu(x,y)s_x$ & $\sum_{x\geq y}\mu(y,x)s_x$ \\
		Inverse DLT $s_x = $ & $\sum_{y\leq x}\ft{s}_y$ & $\sum_{y\geq x}\ft{s}_y$\\
		Frequency resp. $\fr{h}_y=$ & $\sum_{x\geq y} h_x$ & $\sum_{x\leq y} h_x$ \\
		Inverse frequency resp. $h_x=$ & $\sum_{y\geq x}\mu(x,y)\fr{h}_y$ & $\sum_{y\leq x}\mu(y,x)\fr{h}_y$\\
		Total variation & $\TV(\coord{f}^y) = (\chr{y\not\leq g})_{g\in\gens}$ & $\TV(\coord{f}^y) = (\chr{y\not\geq g})_{g\in\gens}$ \\
		Frequency ordering & $x\leq y\Leftrightarrow\TV(\coord{f}^x)\leq \TV(\coord{f}^y)$ & $x\leq y\Leftrightarrow\TV(\coord{f}^x)\geq \TV(\coord{f}^y)$\\
		Lowest frequency & $\coord{f}^\minel = (1,\dots,1)^T$ & $\coord{f}^\maxel = (1,\dots,1)^T$\\
		\bottomrule
	\end{tabular}
	\caption{Basic DLSP concepts for meet- and join-semilattices.\label{summary}}
\end{table*}

\mypar{Inverse frequency response}
The spectrum $\ft{\coord{s}}$ in \eqref{dlt} and the frequency response $\fr{\coord{h}}$ in \eqref{freqresp} are computed differently. However, \eqref{freqresp} is similar to the inverse DLT in \eqref{invdlt} and thus can also be inverted using the Moebius inversion formula, which now takes the form (see \cite[Sec.~3, Cor.~1]{Rota:64})
\begin{equation}\label{invfreqresp}
h_x = \sum_{y\geq x}\mu(x,y)\fr{h}_y,\quad x\in\latt.
\end{equation}
Since it is invertible, every frequency response can be realized by a suitably chosen filter. In particular, the trivial filter exists that leaves every signal unchanged, i.e., with $\fr{h}_y = 1$ for all $y\in\latt$. In other words, there is a linear combination of the $T_a$, $a\in\latt$, that is the identity matrix. If $\latt$ has a unique maximum $\maxel$, then the trivial filter is $T_\maxel$. 

\mypar{Convolution theorem} The preceding results yield the following convolution theorem:
$$
\ft{\coord{h}\ast\coord{s}} = \fr{\coord{h}}\odot\ft{\coord{s}},
$$
where $\odot$ denotes pointwise multiplication.

\mypar{Fast algorithms} Fourier transform, frequency response, and their inverses can be computed in $O(kn)$ many operation where $k = |\gens|$ \cite{Bjorklund:15,Kaski.Kohonen.Westerbaeck:2016a}. In some cases this can be further improved to $O(|\edges|)$, where $|\edges|$ is the number of covering pairs.

\subsection{Intuition}\label{sec:intuition}

We now provide some intuition for the meaning of spectrum and Fourier transform of a signal $\coord{s}$. We call $y$ a predecessor of $x\in\latt$ if $y\leq x$. Suppose there exists an (unknown) weight or cost function on the lattice
\begin{equation}\label{eq:SignedMeasure}
	\begin{split}
		w \colon \latt &\longrightarrow \mathbb{R}, \\
		x &\mapsto w(x) = w_x, 
	\end{split}
\end{equation}
which defines a signed measure \cite[chap.~IX]{Doob:94} on $\latt$, i.e., it can be extended to subsets ${\cal S}$ of $\latt$ as $w({\cal S}) = \sum_{y\in{\cal S}} w_y$. 

Assume that each observed signal value $s_x$ is the sum of the weights of all predecessors
\begin{equation}
	\label{eq:CumulativeSignal}
	s_x = w(\{y\mid y\leq x\}) = \sum_{y \leq x} w_y. 
\end{equation}
Comparing to \eqref{invdlt} shows that $\ft{\coord{s}} = \coord{w}$, i.e., the Fourier transform recovers the unknown weights of elements from the observed signal values. A Fourier coefficient $\ft{s}_x = w_x = 0$ thus implies that
$$
s_x = \sum_{y \leq x} w_y = \sum_{y < x} w_y,
$$
i.e., the signal value is determined already by the weights of the (strictly) smaller elements. If $x$ covers only one element $x'$, i.e., has only one direct predecessor, then $s_x = s_{x'}$ in this case.

Note that $\coord{w}$ cannot be computed through simple differences but the computation is more involved depending on the lattice structure as the formula in \eqref{dlt} indicates (and related to the inclusion-exclusion principle in set theory). In Fig.~\ref{fig:IntuitionFourierSpectrum} we show an example. Note, for instance, that $w_f = s_f - s_d - s_c - s_b + 2 s_a$ since $b,c,d$ have the common predecessor $a$.

The inverse DLT \eqref{eq:SignedMeasure} can be viewed as an integral operator on $\latt$, which makes the DLT the associated differentiation operator \cite[p.~347]{Rota:64}. Indeed, for the totally ordered lattice (e.g., Fig.~\ref{l2}), the DLT reduces to the standard difference operator that computes the difference of successive values.

With this discussion we can also revisit the shift operation in \eqref{shiftdef}. The signal value at index $x\in\latt$, when shifted by $a\in\latt$ is
\begin{equation}
	\label{eq:ShiftIntuition}
	(T_a \coord{s})_x
	= s_{x \meet a}
	= \sum_{y \leq x\text{ and }y \leq a} w(y),
\end{equation}
i.e., it is the sum of the weights of all lattice elements that are predecessors for both $a$ and $x$.

\subsection{Summary and Join-Semilattices}

For convenience we collect the main concepts of DLSP in Table~\ref{summary} together with the dual framework for join-semilattices whose derivation is completely analogous. The two are not fundamentally different but are transposes of each other in the following sense. Every meet-semilattice can be converted into a join-semilattice by inverting the order relation, i.e., the arrows in the cover graph. Meet-irreducible element become join-irreducible this way, and the shift matrices of the latter are transposes of those of the former. As a result, as Table~\ref{summary} shows, the Fourier transform matrices and frequency response matrices are also transposes of each other.

The only major difference is in the frequency ordering, which in the case of a join-semilattice is inverted. This means $x\leq y$ implies that $\TV(\coord{f}^x)\geq \TV(\coord{f}^y)$. In particular, $\coord{f}^\maxel$ is the lowest frequency, which makes sense since it is the constant, all-one vector.

\section{Example}\label{sec:example}

In this section we illustrate and instantiate the DLSP concepts using the small lattice previously shown in Fig.~\ref{l1}.

\mypar{Lattice signal} A lattice signal is shown in Fig.~\ref{exsignal}. It associates a real number $s_x$ with every lattice element~$x$.

\begin{figure*}\centering
\subcaptionbox{Signal on the lattice\label{exsignal}}[0.22\linewidth]
{\includegraphics[scale=0.45]{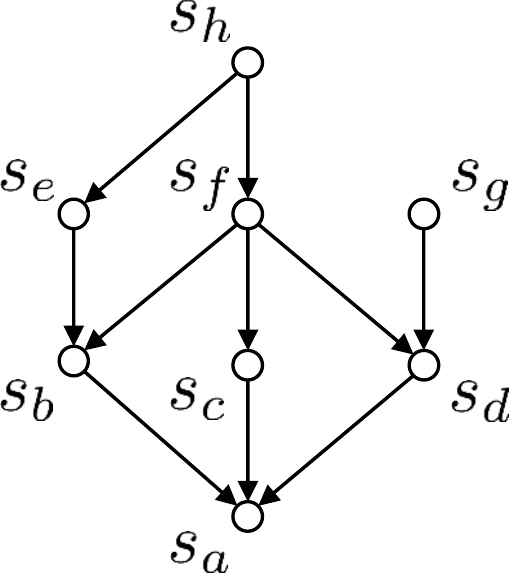}}
\hfill
\subcaptionbox{Basic low pass filter $\coord{h}$\label{exlowpass}}[0.22\linewidth]
{\includegraphics[scale=0.45]{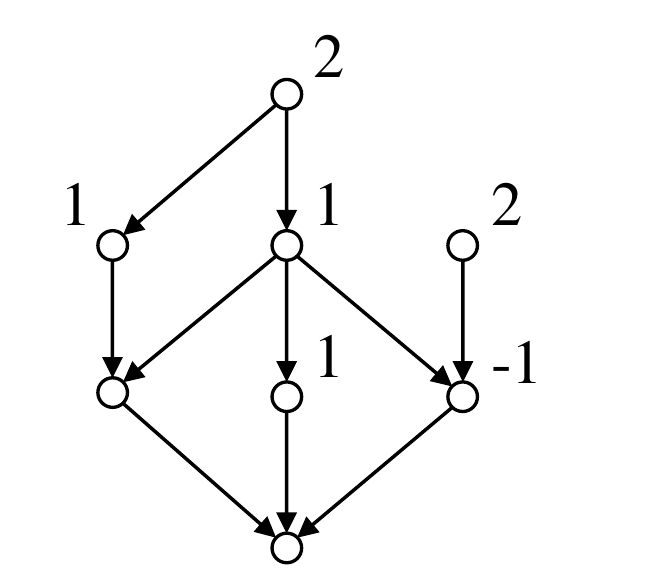}}
\hfill
\subcaptionbox{Frequency response of the filter $\coord{h}$ in (b)\label{exfreqresp}}[0.22\linewidth]
{\includegraphics[scale=0.45]{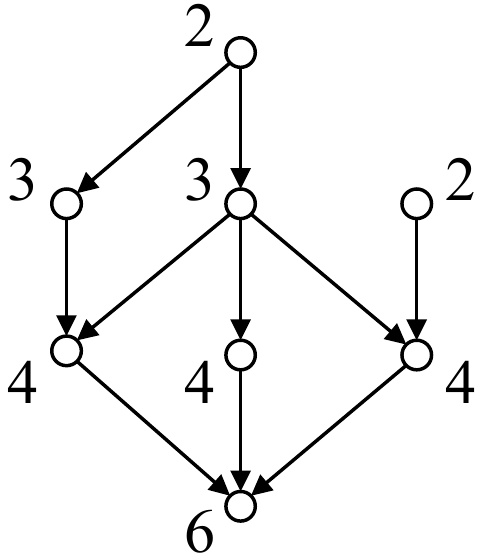}}	
\hfill
\subcaptionbox{Partial ordering of frequencies by total variation\label{extvs}}[0.22\linewidth]
{\includegraphics[scale=0.45]{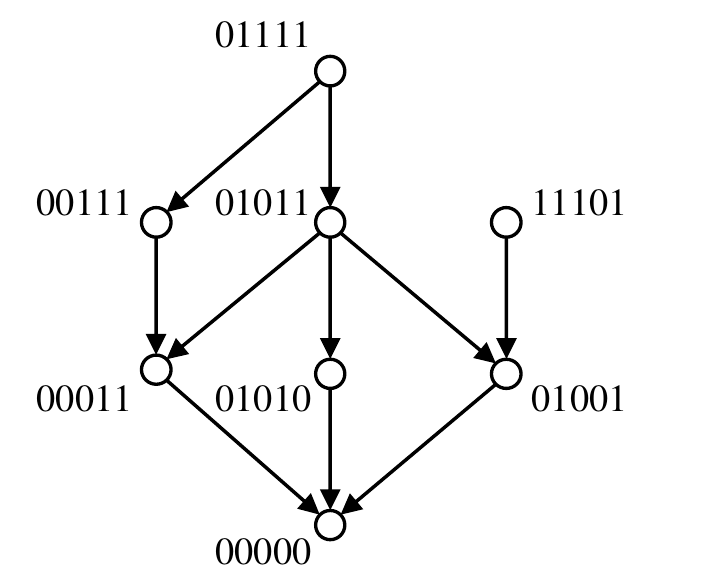}}
\caption{Basic concepts for the example lattice $\latt$ in Fig.~\ref{l1}.\label{exfig}}
\end{figure*}

\mypar{Generating shifts} As mentioned before, the lattice is generated by five shifts: $\gens=\{h,e,f,g,c\}$. For example shifting by $e$ maps
$$
\coord{s}\mapsto T_e\coord{s} = (s_{x\meet e})_{x\in\latt}
= (s_a, s_b, s_a, s_a, s_e, s_b, s_a, s_e)^T,
$$
i.e., $T_e$ is the matrix
$$
T_e = 
\left[
\scriptsize
\ra{1.0}
\begin{array}{*{7}{r@{\quad}}r}
1 & 0 & 0 & 0 & 0 & 0 & 0 & 0\\
0 & 1 & 0 & 0 & 0 & 0 & 0 & 0\\
1 & 0 & 0 & 0 & 0 & 0 & 0 & 0\\
1 & 0 & 0 & 0 & 0 & 0 & 0 & 0\\
0 & 0 & 0 & 0 & 1 & 0 & 0 & 0\\
0 & 1 & 0 & 0 & 0 & 0 & 0 & 0\\
1 & 0 & 0 & 0 & 0 & 0 & 0 & 0\\
0 & 0 & 0 & 0 & 1 & 0 & 0 & 0
\end{array}
\right].
$$
Note that the number of generating shifts $|\gens|$ appears large compared to $|\latt|$. This is due to the small scale of the example. In the subset lattice (Fig.~\ref{l3}), $|\gens| = \log_2(|\latt|)$.

\mypar{Filters and convolution}
A generic filter is given by $\coord{h}=(h_a)_{a\in\latt}$. First we identify the trivial filter $T$ with $T\coord{s} = \coord{s}$ for all $\coord{s}$, i.e., with a constant-one frequency response. Using \eqref{invfreqresp} with $\fr{h}_y=1$ for all $y\in\latt$, we get 
$$
T = T_g + T_h - T_d.
$$

As an example filter, we use what one might expect (and we will confirm below) as a basic low-pass filter (analogous to $h(z) = 1 + z_1^{-1} + z_2^{-1}$ in the $z$-domain for two-dimensional SP) that sums the trivial filter $T$ and all generating shifts (see Fig.~\ref{exlowpass}):
\begin{equation}\label{lowpassmat}
T + \sum_{g\in\gens}T_g = 2T_h + 2T_g + T_f + T_e - T_d + T_c,
\end{equation}
i.e.,
\begin{equation}\label{lowpass}
\coord{h} = (0,0,1,-1,1,1,2,2)^T.
\end{equation}

\mypar{Fourier basis and frequency response}
The Fourier basis vectors are computed with \eqref{purefreq} and are indexed again by $y\in\latt$. Collected in a matrix using the previous ordering of $y\in\latt$, they take the form
\begin{equation}\label{exfreq}
\left[
\scriptsize
\ra{1.0}
\begin{array}{*{8}{r@{\quad}}r}
1 & 0 & 0 & 0 & 0 & 0 & 0 & 0\\
1 & 1 & 0 & 0 & 0 & 0 & 0 & 0\\
1 & 0 & 1 & 0 & 0 & 0 & 0 & 0\\
1 & 0 & 0 & 1 & 0 & 0 & 0 & 0\\
1 & 1 & 0 & 0 & 1 & 0 & 0 & 0\\
1 & 1 & 1 & 1 & 0 & 1 & 0 & 0\\
1 & 0 & 0 & 1 & 0 & 0 & 1 & 0\\
1 & 1 & 1 & 1 & 1 & 1 & 0 & 1
\end{array}
\right].
\end{equation}
The frequency response of $\coord{h}$ in \eqref{lowpass} is computed with \eqref{freqresp} and the result is shown in Fig.~\ref{exfreqresp}. We observe that higher frequencies are attenuated, so indeed $\coord{h}$ is a low-pass filter. We note that the filter would also be low-pass if $T$ in \eqref{lowpassmat} is omitted or if only a subset of the coefficients $h_g,g\in\gens$ would be set to 1.

\mypar{Frequency ordering}
The total variation is a five-tuple for this case, one element per generator $\in\gens$. As explained, these five-tuples form a lattice isomorphic to $\latt$, shown in Fig.~\ref{extvs}, partially ordered by componentwise comparison of five-tuples. Smaller elements have a lower sum total variation. The smallest is 0 for the constant Fourier basis vector $\coord{f}^a = \coord{f}^\minel$, and the largest is 4 for the two maximal elements.

\mypar{Fourier transform} The DLT in matrix form is obtained via the Moebius inversion formula \eqref{dlt} or by inverting the matrix in \eqref{exfreq}:
\begin{equation}\label{exft}
\DLT_\latt =
\left[
\scriptsize
\ra{1.0}
\begin{array}{*{8}{r@{\quad}}r}
 1 &  0 &  0 &  0 &  0 &  0 &  0 &  0\\
-1 &  1 &  0 &  0 &  0 &  0 &  0 &  0\\
-1 &  0 &  1 &  0 &  0 &  0 &  0 &  0\\
-1 &  0 &  0 &  1 &  0 &  0 &  0 &  0\\
 0 & -1 &  0 &  0 &  1 &  0 &  0 &  0\\
 2 & -1 & -1 & -1 &  0 &  1 &  0 &  0\\
 0 &  0 &  0 & -1 &  0 &  0 &  1 &  0\\
 0 &  1 &  0 &  0 & -1 & -1 &  0 &  1
\end{array}
\right]
\end{equation}
It diagonalizes the matrix representation of every filter and every shift. For example
$$
\DLT_\latt\cdot T_e\cdot\DLT_\latt^{-1} = \diag(1,1,0,0,1,0,0,0),
$$
and for the low-pass filter $\coord{h}$ in \eqref{lowpass},
$$
\DLT_\latt\cdot\left(\sum_{a\in\latt}h_aT_a\right)\cdot\DLT_\latt^{-1} =
\diag(6,4,4,2,3,3,2,2).
$$
Computing the spectrum of a signal $\coord{s}$ by multiplying with $\DLT_\latt$ corresponds to Fig.~\ref{fig:IntuitionFourierSpectrum}(d).

\section{Properties and Relation to GSP}\label{sec:discussion}

We briefly summarize salient properties of DLSP and discuss the difference to GSP.

\mypar{Properties}
DLSP is combinatorial in nature in that it captures the partial order structure of the index domain. The combinatorial structure shows in the pure frequencies and their total variations, which both consists of 0 and 1 entries only.

The DLT is triangular and thus not orthogonal, so no Parseval identity holds in DLSP. This is not uncommon in SP where so-called biorthogonal transforms have been studied and used in the literature.

A meet-semilattice does not have a unique maximum in general. There are at least three arguments for adding an artificial such maximum. First, adding such an element (if not yet available) makes $\latt$ a lattice, i.e., a simultaneous meet- and join-semilattice. This means that two types of shifts, and thus convolutions, and DLTs are then available. Mixing these into one SP is nontrivial and a possible topic for future research. In particular, meet- and join-shifts cannot be diagonalized simultaneously since they do not commute. Second, if $\latt$ contains a unique maximum $\maxel$, the trivial filter is simply given by $T_\maxel$, which is the identity matrix. Finally, we will explain in Section~\ref{sec:subsetlatt} that by adding a maximum we can encode lattice elements as bit vectors in a way that the bitwise-and corresponds to the meet operation.

It is intriguing that DLSP provides a {\em partial} ordering of frequencies. However, the same occurs in two-dimensional separable SP based on the separable DFT. Applied to an image, for example, yields a two-dimensional array of frequencies that can also be viewed as partially ordered.

Note that the Fourier transform and frequency response are calculated differently. This is not surprising as they have different algebraic roles \cite{Pueschel:08a} and this also happens in graph SP.

\mypar{Difference to GSP} DLSP is different from graph SP in several fundamental ways. 

First, lattices, if represented as cover graphs, constitute a very special subclass of directed acyclic graphs. Note that directed acyclic graphs are troublesome in graph SP since the adjacency matrix has only
the eigenvalue~0 (assuming no self-loops) and is not (or rather far from) diagonalizable with large Jordan blocks. DLSP, in contrast, provides a proper Fourier eigenbasis.

The reason is a very different notion of shift. In GSP, the commonly used Laplacian or adjacency shift move values to or from the nearest neighbors. The DLSP shifts, in that sense, operate at larger distances, as determined by the lattice structure. For example, shifting the value $s_e$ in Fig.~\ref{exsignal} by $g$ yields $s_{e\meet g} = s_a$. However, the shifts still obey a notion of nearest. Namely, assume we say that $x$ is a predecessor of $y$ if $x\leq y$, then the index of $s_{e\meet g}$ is the nearest predecessor of both $e$ and $g$.

Further, unlike in GSP and as consequence of the lattice structure, DLSP is based on not one but several basic (generating) shifts that are needed to generate the filter algebra. In fact, every linear SP framework based on one generating shift can be viewed as a form of graph SP \cite[p.~56]{Pueschel:06c} and vice-versa.

One may wonder, which shift or shifts are the best one to choose. There is no simple answer as even in GSP a number of different shifts have been proposed (e.g., the six listed in \cite[Table I]{Anis:16}). Each choice yields a different choice of associated Fourier basis to represent signals and thus can reveal different forms of structure like sparsity or uneven energy distribution in the spectrum.

Finally, we note that an advantageous feature of DLSP is that the Fourier basis and Fourier transform have closed forms and thus do not require an eigendecomposition for their construction. This should enable the scaling to lattices with millions of elements; Fourier analysis with GSP is to date restricted to graphs with only thousands of nodes.

\mypar{Other related work} Lattice theory \cite{Graetzer:11} is well-developed but does not discuss convolution or Fourier analysis. Fast Fourier transforms on groups and other algebraic structures are a classical topic in computer science. Most closely related our work in this domain is \cite{Bjorklund:15} (and some of the references therein). The authors derive fast lattice Fourier transform algorithms in the sense introduced here, but not a complete SP framework. 

The simplicial complexes considered as index domain in \cite{Barbarossa:20} are meet-semilattices but possess additional topological structure that the authors use to define an SP framework very different from ours. 

We show later in Section~\ref{sec:FormalConceptLattices} that DLSP can be applied to hypergraphs, but in an indirect way that is different from the tensor-based approach in \cite{Zhang:20}.

The powerset is a lattice with a directed hypercube as cover graph. Powerset signals, usually called set functions, have a rich set of applications in machine learning and beyond \cite{Krause:14}. We introduced an SP framework for set functions in \cite{Pueschel:20}, part of which can be seen as special case of this paper. However, the special structure of powersets provide additional shifts (namely, the set difference has no generalization to lattices) and thus SP frameworks and an interpretation as so-called coverage functions. 

Finally, \cite{Riess.Hansen:2020} used the lattice convolution from our preliminary \cite{Pueschel:19} to define lattice neural networks.

\section{Sampling Theorem}\label{sec:sampling}

We now extend DLSP with a sampling
theorem, first presented in~\cite{Wendler:19}, for the perfect
reconstruction of lattice signals that are sparse in the DLT Fourier
domain, i.e., where only a subset of the spectrum is non-zero.

We call a lattice signal $\coord{s}$ $k$-Fourier-sparse if its Fourier support $\ftsu = \text{supp}(\hat{\coord{s}})$ 
satisfies
\begin{equation}
	\label{eq:DefFourierSparse}
	|\ftsu| = |\text{supp}(\hat{\coord{s}})| = |\{b \in\latt: \hat s_b
	\neq 0\}| = k.     
\end{equation}
As an example, the signal in Fig.~\ref{exsig} is 4-Fourier-sparse as shown in Fig.~\ref{exsigfour}. Following the general paradigm of classical sampling
theory~\cite{Vetterli:14}, we are looking for a linear
sampling operator $P_{\sisu}$ that reduces the signal $\coord{s}$ to
$|\sisu| = k = |\ftsu|$ samples such that there exists a linear
interpolation operator $I_\sisu$ that allows for perfect recovery of
$\coord{s}$ from these samples:
$\coord{s} = I_\sisu P_\sisu\coord{s}$. The triangular shape of the
Fourier transform makes this task easy.

\begin{figure*}\centering
	\subcaptionbox{Example signal $\coord{s}$\label{exsig}}[0.24\linewidth]
	{\includegraphics[width=0.6\linewidth]{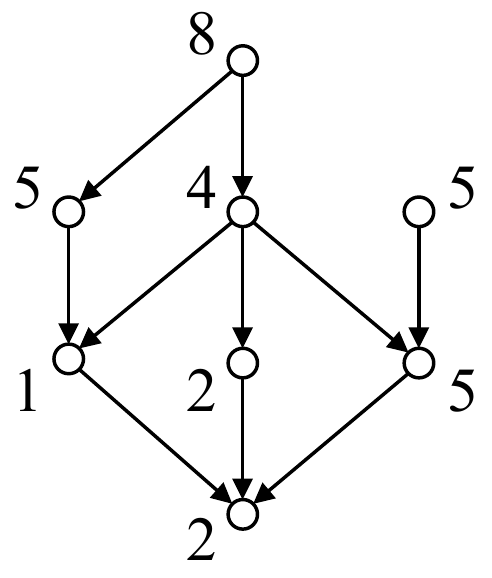}}
	\hfill
	\subcaptionbox{Its spectrum $\ft{\coord{s}}$\label{exsigfour}}[0.24\linewidth]
	{\includegraphics[width=0.6\linewidth]{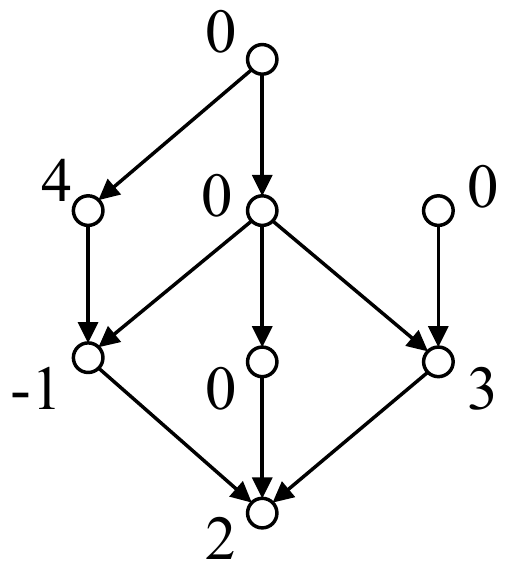}}
	\hfill
	\subcaptionbox{$P_B \coord{s} = \coord{s}_B$\label{ex:sampling}}[0.24\linewidth] 
	{
		\resizebox{40mm}{!}{
			$	
			\left[
			\begin{array}{rrrr}
				1 & 0 & 0 & 0 \\
				0 & 1 & 0 & 0 \\
				0 & 0 & 0 & 0 \\
				0 & 0 & 1 & 0 \\
				0 & 0 & 0 & 1 \\
				0 & 0 & 0 & 0 \\
				0 & 0 & 0 & 0 \\
				0 & 0 & 0 & 0 \\
			\end{array}
			\right]^{T}
			\cdot
			\left[
			\begin{array}{r}
				2 \\
				1 \\
				2 \\
				5 \\
				5 \\
				4 \\
				5 \\
				8 \\
			\end{array}
			\right]
			=
			\left[
			\begin{array}{r}
				2 \\
				1 \\
				5 \\
				5 \\
			\end{array}
			\right]
			$
	}}
	\hfill
	\subcaptionbox{$\coord{s} = I_B \coord{s}_B$ \label{ex:interpolation}}[0.24\linewidth]
	{\resizebox{40mm}{!}{
			$
			\left[
			\begin{array}{rrrr}
				1 & 0 & 0 & 0 \\
				0 & 1 & 0 & 0 \\
				1 & 0 & 0 & 0 \\
				0 & 0 & 1 & 0 \\
				0 & 0 & 0 & 1 \\
				-1 & 1 & 1 & 0 \\
				0 & 0 & 1 & 0 \\
				-1 & 0 & 1 & 1 \\                  
			\end{array}
			\right]
			\cdot
			\left[
			\begin{array}{r}
				2 \\
				1 \\
				5 \\
				5 \\
			\end{array}
			\right]
			=
			\left[
			\begin{array}{r}
				2 \\
				1 \\
				2 \\
				5 \\
				5 \\
				4 \\
				5 \\
				8 \\
			\end{array}
			\right]
			$
	}}
	\caption{The signal in (a) on $\latt$ in Fig.~\ref{l1} is
		4-Fourier-sparse with support
		$\ftsu = \{a, b, d, f\}$ as shown in (b). The sampling in (c) followed by the interpolation in (d) allows for perfect reconstruction.\label{fig:ExampleSamplingAndInterpolation}}
\end{figure*}

Let $\sisu = \{a_1,\dots,a_k\}\subseteq\latt$. We consider the
associated sampling operator
\begin{equation}
	\label{eq:LatticeSamplingOperator}%
	P_\sisu: \R^n\mapsto\R^k,\
	\coord{s} \mapsto \coord{s}_{\sisu} = (s_{a_1}, \dots, s_{a_k})^T.
\end{equation}

\begin{theorem}[Lattice sampling]
	\label{thm:sampling}
	Let $\coord{s}$ be a $k$-sparse lattice signal on $\latt$ with
	Fourier support $\ftsu$. Then $\coord{s}$ can be perfectly
	reconstructed from the samples
	$\coord{s}_\ftsu = P_\ftsu\coord{s}$.
	
	Namely, $\coord{s} = I_\ftsu \coord{s}_\ftsu$ with
	$I_\ftsu = \DLT^{-1}_{\latt,\ftsu}
	(\DLT^{-1}_{\ftsu,\ftsu})^{-1}$. The matrix
	$\DLT^{-1}_{\sisu, \ftsu}$ is the sub-matrix of $\DLT^{-1}$
	obtained by selecting the rows indexed by $\sisu$ and the columns
	indexed by $\ftsu$.
\end{theorem}
\begin{proof}
	As $\coord{s}$ has Fourier support $\ftsu$, we have
	\begin{equation*}
		\coord{s} = \DLT^{-1}_{\latt,\ftsu} \hat{\coord{s}}_\ftsu.
	\end{equation*}
	Applying the sampling operator $P_\ftsu$ to both sides yields
	\begin{equation*}
		\coord{s}_\ftsu = \DLT^{-1}_{\ftsu,\ftsu} \hat{\coord{s}}_\ftsu.
	\end{equation*}
	What remains to show is that $\DLT^{-1}_{\ftsu,\ftsu}$ is
	invertible. However, this is the case because of its triangular
	shape (when $\latt$ is topologically sorted) with 1s on the
	diagonal.
\end{proof} 

\mypar{Example} Fig.~\ref{fig:ExampleSamplingAndInterpolation} instantiates 
Theorem~\ref{thm:sampling} for the 4-Fourier-sparse signal in Fig.~\ref{exsig}. In this case, the sampling operator samples the signal at indices $\{a,b,d,f\}$ as shown in Fig.~\ref{ex:sampling}. The corresponding interpolation operator in Fig.~\ref{ex:interpolation} enables perfect reconstruction.

\section{Subset Lattices}\label{sec:subsetlatt}

The powerset $2^N$ of a finite set $N$ forms a lattice with meet $\cap$ and join $\cup$ and its cover graph is a hypercube (e.g., Fig.~\ref{l3}). Signals on the powerset lattice are called set functions and the Fourier basis and Fourier transform have a very regular structure and possess additional properties \cite{Pueschel:20}.

Any subset of the powerset lattice that is closed under $\cap$ is a meet-semilattice. In the following we show that the reverse is also true, i.e., every meet-semilattice can be viewed as a subset lattice with $\cap$ as meet. For example, the lattice in Fig.~\ref{l2} is isomorphic to the subset lattice $\{a,c,e,h\}$ in Fig.~\ref{l3}. This is relevant since it enables the encoding of lattice elements, viewed as subsets, as bit vectors with bitwise-and as the meet operation. Further, it shows that lattice signals can be viewed as sparse set functions.

Let $\latt$ be a meet-semilattice. If $\latt$ does not contain a unique maximum, we add one $\latt'= \latt\cup\{\maxel\}$ to obtain a lattice (Lemma~\ref{addmax}). In $\latt'$ we let $\gens$ be the set of join-irreducible elements. Note that $\maxel\not\in\gens$ since otherwise it would cover only one element which would then already be a unique maximum in $\latt$. Thus, $\gens\subseteq\latt$. With this we can express $\latt$ isomorphically as a subset lattice:

\begin{theorem}\label{embed}
The mapping
$$
\phi:\ \latt\rightarrow 2^\gens,\quad x\mapsto \{g\in\gens\mid g\leq x\}
$$
satisfies
$$
\phi(x\meet y) = \phi(x)\cap\phi(y).
$$
\end{theorem}
\begin{proof}
$\phi(x\meet y) = \{g\in\gens\mid g\leq x\meet y\} = \{g\in\gens\mid g\leq x\text{ and }g\leq y\} = \phi(x)\cap\phi(y)$.
\end{proof}
The analogous construction for a join-semilattice yields an isomorphism that satisfies $\phi(x\join y) = \phi(x)\cap\phi(y)$. 

For the lattice in Fig.~\ref{l1}, we obtain the join-irreducibles $\gens = \{b,c,d,e,g\}$ and thus $\latt$ is isomorphic to a subset lattice of a powerset lattice of a set with five elements.

The construction in Theorem~\ref{embed} was also used in \cite{Bjorklund:15} to obtain fast lattice Fourier transforms.


\section{Application Example: Formal Concept Lattices in Social Data Analysis}
\label{sec:FormalConceptLattices}%

In this section we present one possible application domain for DLSP: the well-developed area of formal concept lattices (FCLs) used in social data analysis~\cite{Ganter.Wille:2012a}. These lattices can be seen as a representation of relations between objects and attributes. Relations can be viewed equivalently as bipartite graphs or hypergraphs. DLSP provides a notion of filtering and Fourier analysis for data on relations via FCLs.

We first introduce the notion of FCLs using a small example. Then we provide a formal definition of FCLs and consider a larger dataset on information about customers of a telecommunication company.

\subsection{A Small Motivating Example}\label{sec:smallfcl}

Table~\ref{table:SmallFCADataTelco} shows a sample of seven customers (users) of a telecommunication
service with eight binary properties and the additional churn
  variable, which indicates whether a user did cancel her or his
  contract\footnote{The data is part of the IBM sample datasets
  \url{https://www.ibm.com/support/knowledgecenter/SSEP7J_11.1.0/com.ibm.swg.ba.cognos.ig_smples.doc/c_telco_dm_sam.html}.
  The dataset is available under
  \url{https://www.kaggle.com/blastchar/telco-customer-churn}}. For
example, user $U_3$ has internet service, prefers paperless
billing and canceled the contract.

\newcommand{\TableOne}{$\times$}
\newcommand{\TableZero}{}
\newcommand{\MyGray}{gray!30}
\begin{table}
    \ra{1}
    \centering
    \small
    \begingroup\fboxsep=3pt
    \begin{tabular}{@{}l*{3}{c@{\,}}*{2}{c@{}}*{2}{c@{\,}}@{}}\toprule
      & \multicolumn{7}{c}{User}  \\
      \cmidrule{2-8}
      Property & $U_1$ & $U_2$ &  $U_3$ &  $U_4$ &  $U_5$ & $U_6$ & $U_7$ \\
      \midrule $P_1$: Gender  & \TableZero & \TableOne & \TableZero &{\colorbox{\MyGray}{\TableOne}}&\colorbox{\MyGray}{\TableOne}& \TableZero & \TableZero\\
      $P_2$: Partner & \TableOne & \TableZero & \TableZero &\colorbox{\MyGray}{\TableOne}&\colorbox{\MyGray}{\TableOne}& \TableOne & \TableOne \\
      $P_3$: Dependents & \TableZero & \TableZero & \TableZero &\colorbox{\MyGray}{\TableOne}&\colorbox{\MyGray}{\TableOne}& \TableOne & \TableOne\\
      $P_4$: InternetService &\TableOne & \TableOne & \TableOne & \TableZero & \TableOne & \TableZero & \TableOne\\
      $P_5$: DeviceProtection & \TableZero & \TableOne & \TableZero & \TableZero & \TableOne & \TableZero & \TableZero \\
      $P_6$: TechSupport & \TableZero & \TableZero & \TableZero & \TableZero & \TableZero & \TableZero & \TableOne\\
      $P_7$: StreamingMovies & \TableZero & \TableOne &\TableZero & \TableZero & \TableOne & \TableZero & \TableZero\\
      $P_8$: PaperlessBilling & \TableOne & \TableOne & \TableOne & \TableZero & \TableOne & \TableZero & \TableZero \\
      \midrule Churn & 1 &  0 & 1 & 0 & 1 & 0 & 1 \\
	 \bottomrule
    \end{tabular}
    \endgroup
\caption{A sample of a telecommunications dataset. One formal concept is highlighted.\label{table:SmallFCADataTelco}} 
\end{table}

Every user satisfies a number of properties and every property is satisfied by a certain number of users. This relationship can be extended to sets of users and sets of properties. For example properties $\{P_1,P_2\}$ are satisfied by the users $\{U_4,U_5\}$. Conversely, the properties jointly satisfied by $\{U_4,U_5\}$ also include $P_3$: $\{P_1,P_2,P_3\}$, so this relationship between sets of users and sets of properties is not always one-to-one.
For $\{U_4,U_5\}$ and $\{P_1,P_2,P_3\}$ it is and defines a so-called formal concept.

More specifically, a formal concept is a pair of a set of objects (here: users) and a set of attributes (here: properties) such that the objects are uniquely specified by the attributes and vice versa. In Table~\ref{table:SmallFCADataTelco} each formal concept corresponds to a maximal rectangle of crosses. The example concept above is highlighted. Note that a maximal rectangle does not need to be contiguous (e.g., $(\{U_2,U_5\}, \{P_4,P_5,P_7,P_8\})$).

The formal concepts form the FCL, which for Table~\ref{table:SmallFCADataTelco} is shown in Fig.~\ref{fig:SmallTelcoFCALattice}, using a shorthand notation for sets. The formal concepts are partially ordered by inclusion of object sets, or, equivalently, by inverse inclusion of attribute sets (more objects mean that fewer attributes are satisfied). The meet is obtained by intersecting attribute sets and the join by intersecting object sets.

For example, the minimal element $(U_{12\cdots 7},\emptyset)$ in Fig.~\ref{fig:SmallTelcoFCALattice} shows that no property is satisfied by all users and vice-versa. The aforementioned $(U_{45}, P_{123})$ is a lattice element. The meet of $(U_7,P_{2346})$ and $(U_{45},P_{123})$ is obtained by intersecting the property sets and yields $(U_{4567},P_{23})$.

\begin{figure}
    \centering
    \includegraphics[width=0.7\linewidth]{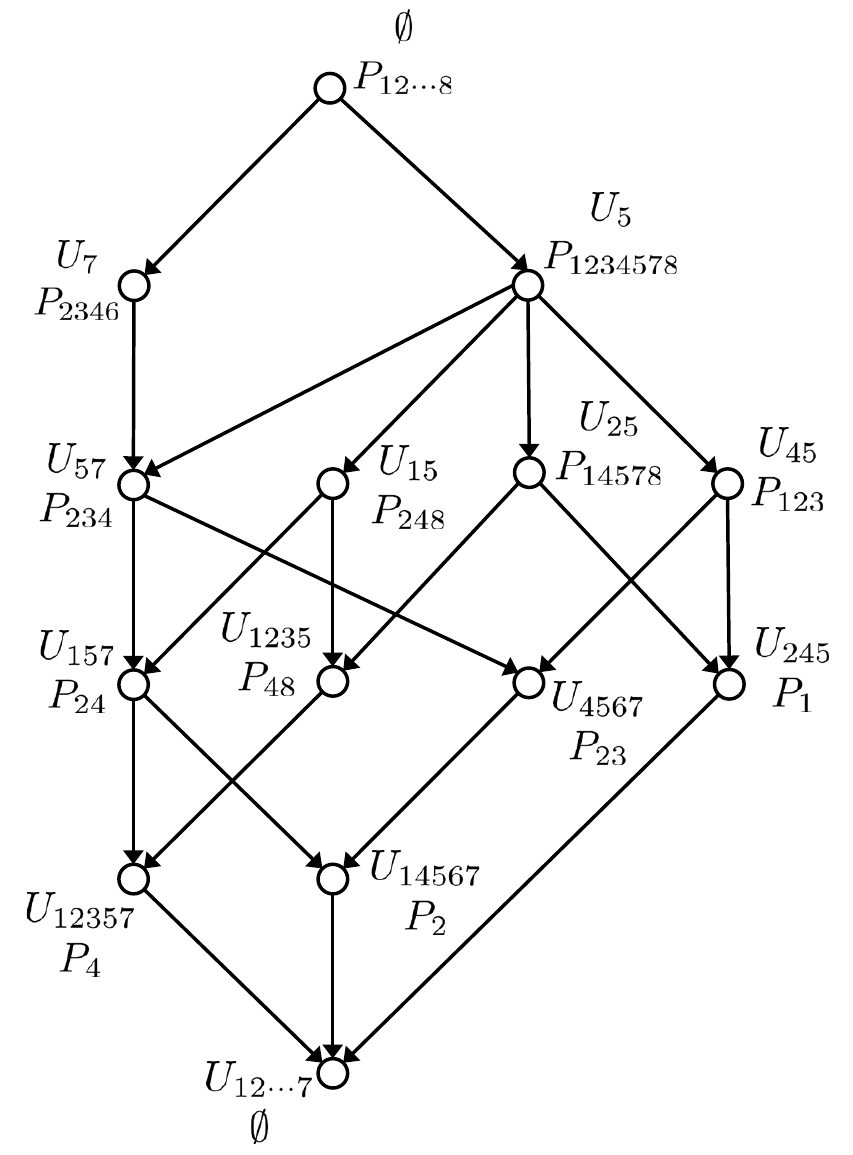}
    \caption{The formal concept lattice, with 14 elements, of the dataset in
      Table~\ref{table:SmallFCADataTelco}. We use the shorthand
      notations $P_{k_1,\ldots,k_j} = \{P_{k_1}, \ldots, P_{k_j}\}$
      and $U_{k_1,\ldots,k_j} = \{U_{k_1},\ldots,U_{k_j}\}$.}
    \label{fig:SmallTelcoFCALattice}
\end{figure}

Finally, we note that the original Table~\ref{table:SmallFCADataTelco} can be reconstructed from the FCL, since every {\TableOne} is contained in at least one formal concept.

Assume that we want to analyze the mapping from user
properties to the churn, i.e., whether users have canceled their
contract. Users with the same properties are indistinguishable in this mapping and thus it makes sense to collect them into sets as done in the FCL.

We define as signal on the FCL the average churn of all users within a
formal concept. In Fig.~\ref{fig:SmallTelcoFCASignal} we show the
obtained signal together with its meet and join Fourier spectrum.

\begin{figure*}
	\centering
	\subcaptionbox{Lattice (properties only)\label{subfig:SmalltelcoSignal}}[0.24\linewidth]{
		\includegraphics[width=0.93\linewidth]{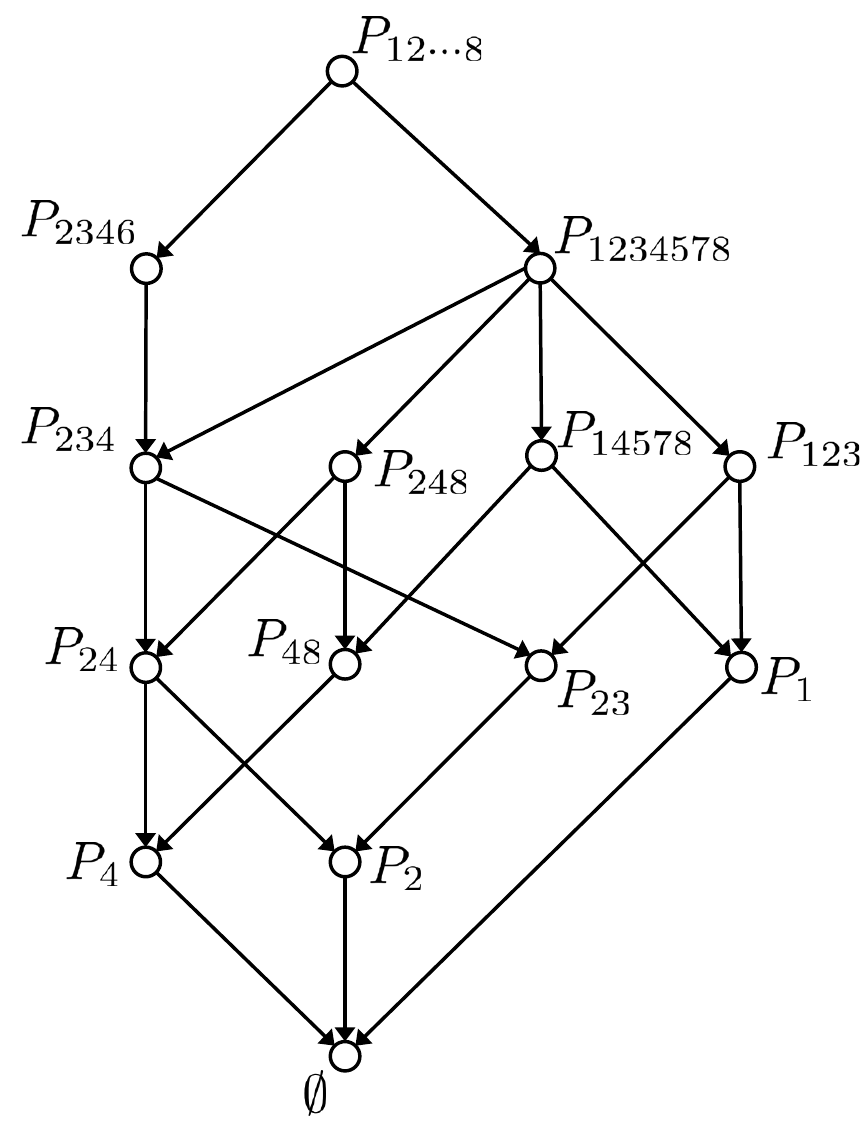}
	}
	\subcaptionbox{Lattice signal\label{subfig:SmalltelcoSignal}}[0.24\linewidth]{
		\includegraphics[width=0.8\linewidth]{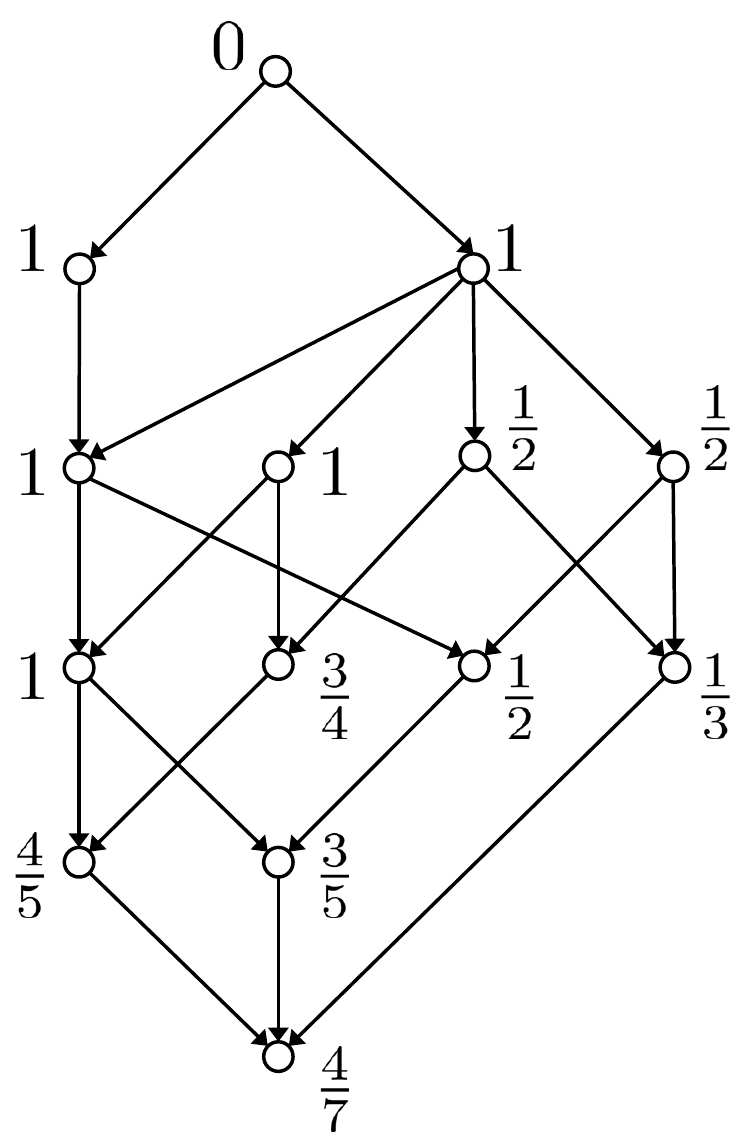}
	}
	\subcaptionbox{Meet spectrum\label{subfig:SmalltelcoSignalMeet}}[0.24\linewidth]{
		\includegraphics[width=0.9\linewidth]{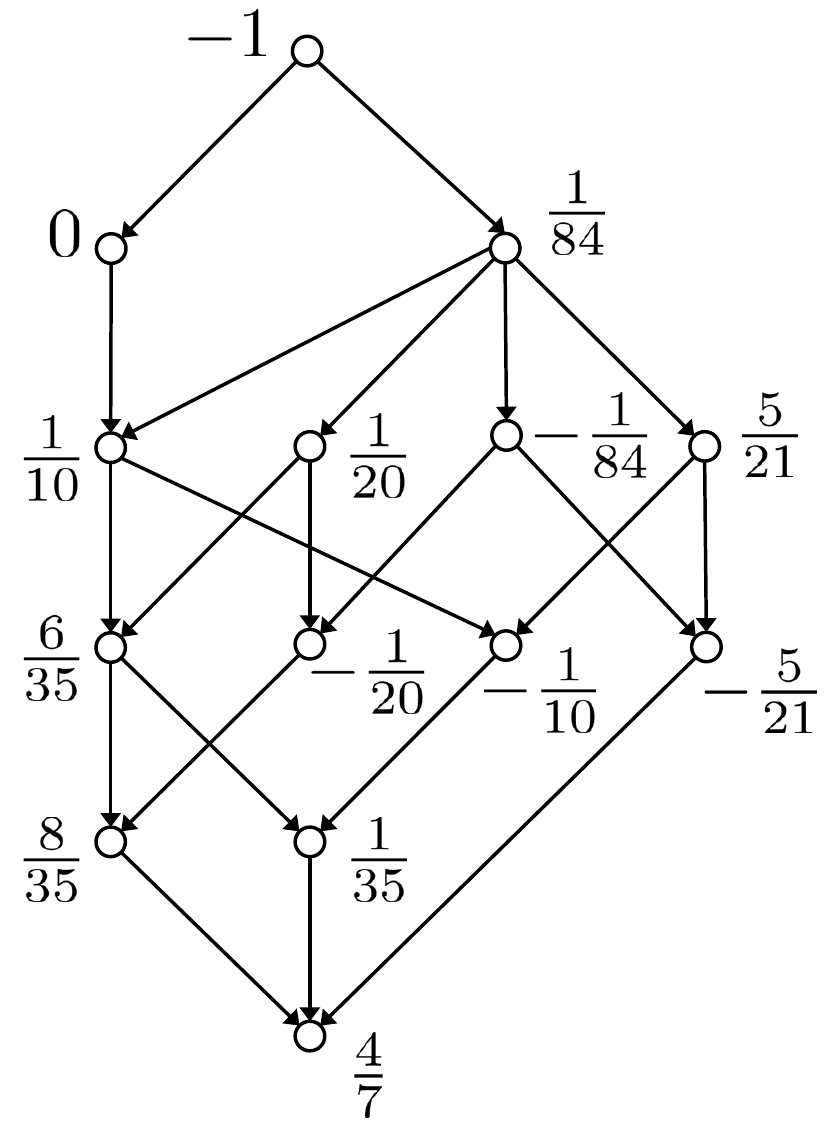}
	}
	\subcaptionbox{Join spectrum\label{subfig:SmalltelcoSignalJoin}}[0.24\linewidth]{
		\includegraphics[width=0.9\linewidth]{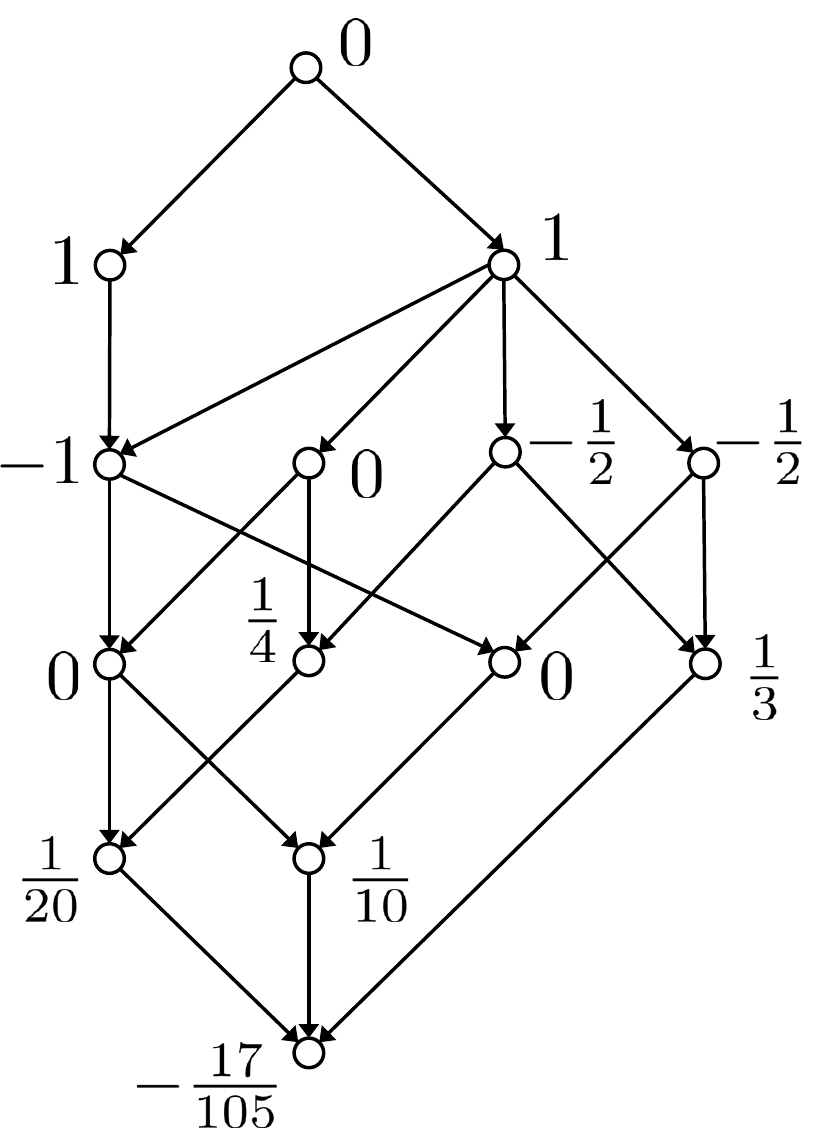}
	}
	\hfill
	\caption{The signal and its meet and join spectrum on the formal
		concept lattice of the data in Table~\ref{table:SmallFCADataTelco}.}
	\label{fig:SmallTelcoFCASignal}
\end{figure*}

For example, Fig.~\ref{fig:SmallTelcoFCASignal}(b) shows that the value at the minimum (no properties, satisfied by all users) is $4/7$, the average over all users. The value for $P_{48}$ (properties 4 and 8, satisfied by users 1,2,3,5) is $3/4$, since 3 of these users have a churn of 1.

Using the intuition provided in Section~\ref{sec:intuition}, the meet spectrum, in the partial order
sense, captures the changes to the churn when adding properties to
concepts (from bottom to top). If a node has only one direct predecessor, it is simply the difference. E.g., the meet spectral value at $P_4$ is $8/35 = 4/5 - 4/7$, the difference of the signal values at $P_4$ and $\emptyset$. If a node has more than one predecessors (e.g., $P_{248}$) it is more complicated as shown by the DLT formula \eqref{dlt}.

Analogously, the join spectrum captures the changes to the churn when removing properties from
concepts (from top to bottom). Again, if a node has only one predecessor (here $y$ is a direct predecessor of $x$ if $y$ covers $x$) it is the difference, otherwise a more involved computation.

The signal is Fourier-sparse in the join spectrum. In Section~\ref{sec:largefcl} we will show for an equivalent, but much larger dataset how both spectra are sparse and reveal structure.

\subsection{Formal Concept Lattice} 

We now define formal concept lattices formally. Let ${\cal O}$ be a finite set of objects and ${\cal A}$ a
finite set of attributes. Let ${\cal R} \subseteq {\cal O} \times {\cal A}$ be a binary
relation with $(o, a) \in {\cal R}$ iff object $o$ has attribute $a$. For
$O \subseteq {\cal O}$ and $A \subseteq {\cal A}$, let
\begin{equation*}
    \attr(O) = \{a \in {\cal A}\mid (o,a) \in {\cal R} \text{ for all } o \in O\}
\end{equation*}
be the set of joint attributes of the objects in $O$, and let
\begin{equation*}
    \obj(A) = \{o \in {\cal O} \; | \;  (o,a) \in I \text{ for all } a \in A\}
\end{equation*}
be the set of objects that satisfy all attributes in $A$. 

A pair $(O, A) \in 2^{\cal O} \times 2^{\cal A}$ is called a {\em formal concept} if $\attr(O) = A$ and $\obj(A) = O$. In a relation table like Table~\ref{table:SmallFCADataTelco} these correspond to the maximal (not necessarily contiguous) rectangles of crosses. 

The formal concepts induced by the relation ${\cal R}$ form a lattice with partial order $(O_1, A_1) \leq (O_2, A_2)$ iff $A_1 \subseteq A_2$ (or equivalently iff $O_2 \subseteq O_1$). The meet operation is given by
\begin{multline}\label{eq:MeetFormalConcepts}
    (O_1, A_1) \meet (O_2, A_2) = \\(\obj(\attr(O_1 \cup O_2)), A_1 \cap A_2),    
\end{multline}
i.e., by intersecting attributes and collecting the associated objects. In the relation table it corresponds to the maximal rectangle which contains the attributes associated with the two initial rectangles.

Analogously, the join is given by
\begin{multline}
    \label{eq:JoinFormalConcepts}
    (O_1, A_1) \join (O_2, A_2) = \\(O_1 \cap O_2, \attr(\obj(A_1 \cup A_2))).
\end{multline}
In the relation table the join corresponds to the maximal rectangle which contains the properties associated with the two initial rectangles.

We note that, dually, the FCL can be defined to be upside down, i.e., with inverted order and flipped meet/join definitions. We also note that every lattice is isomorphic to a suitable FCL~\cite[Thm.~1]{Gantner.Wille:1997a}.


\subsection{Large Scale Example}\label{sec:largefcl}

We now expand the small example in Section~\ref{sec:smallfcl} to the complete telecommunication dataset with 7043 customers and 11 binary properties. Using the Python library \emph{Concepts}\footnote{Available online at \url{https://github.com/xflr6/concepts}.} we calculate the associated FCL, which has 813 elements. As lattice signal we again use the average of the churn variable for a user set in a concept.

The signal and its meet and join Fourier transforms are shown in
Fig.~\ref{fig:telcoMiddle}. The colors for the signal values are chosen to display the deviation to the mean churn over all customers, which is $1869/7043\approx 0.265$ and shown as white. In the spectra, white means zero. Note that the lattice for the join spectrum is drawn upside down such that low frequencies are at the bottom in both spectra (see last two rows in Table~\ref{summary}).

\begin{figure*}
    \centering
    \hfill
    \subcaptionbox{Lattice signal\label{subfig:telcoSignal}}[0.3\linewidth]{
      \includegraphics[width=\linewidth]{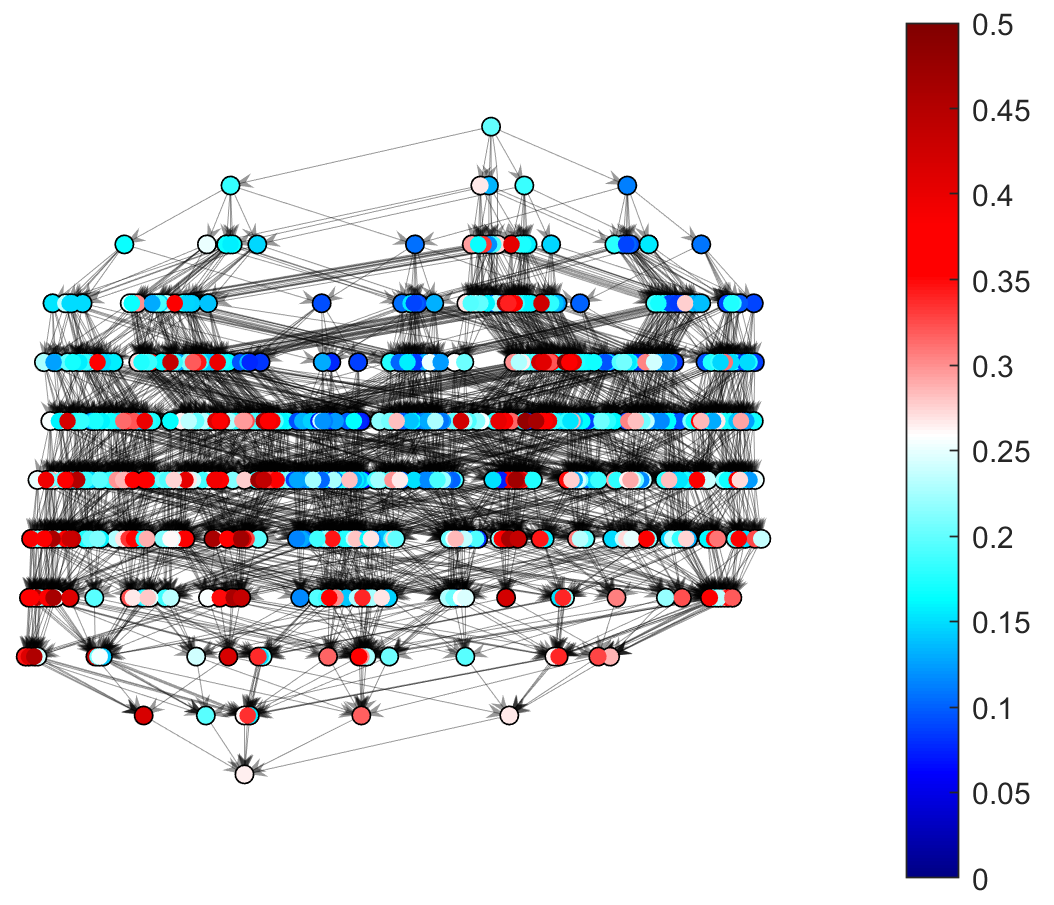}
    }
    \hfill
    \subcaptionbox{Meet spectrum\label{subfig:telcoMeetFourier}}[0.3\linewidth]{
      \includegraphics[width=\linewidth]{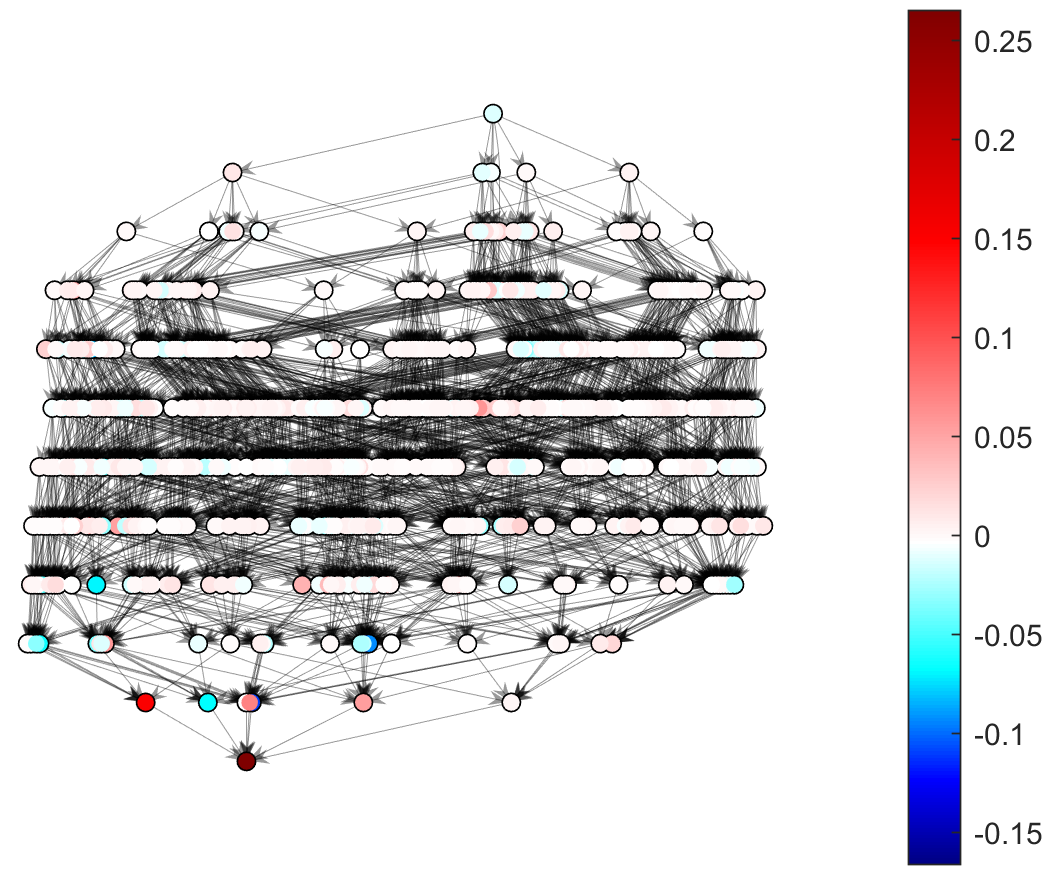}
    }
    \hfill
    \subcaptionbox{Join spectrum\label{subfig:telcoJoinFourier}}[0.3\linewidth]{
      \includegraphics[width=\linewidth]{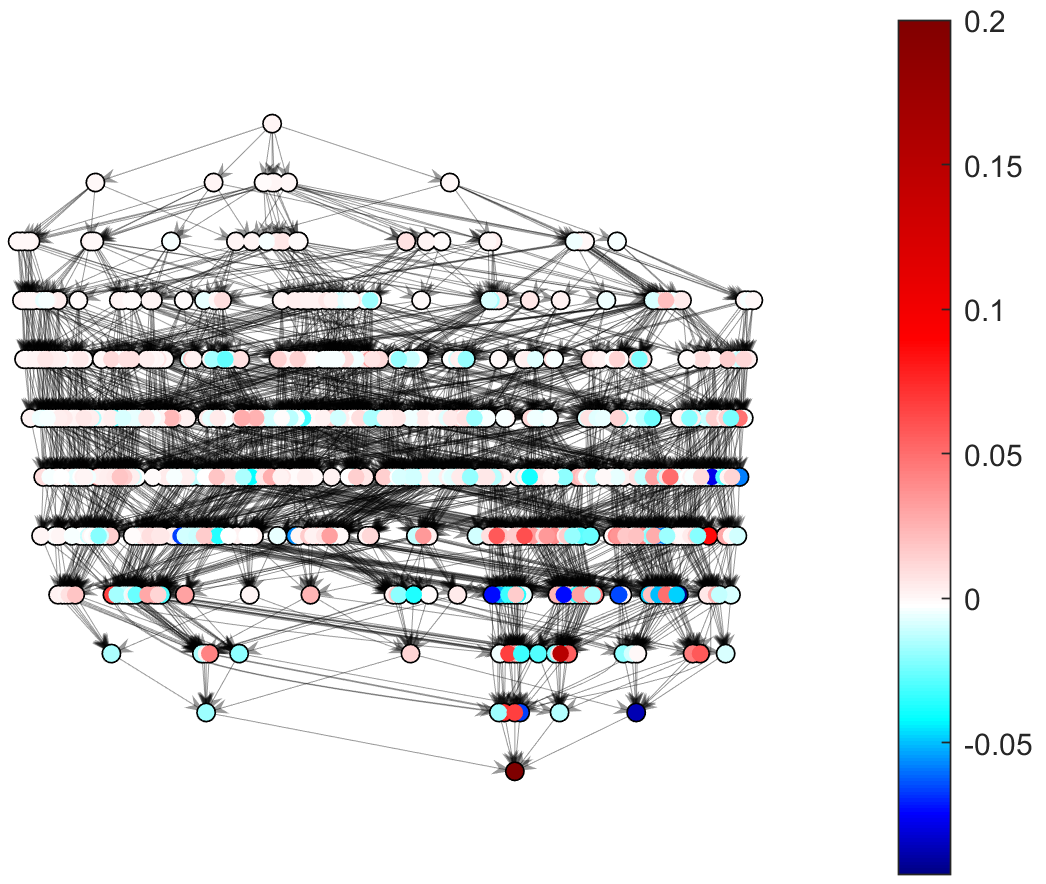}
    }
    \hfill
    \caption{The signal and its meet and join spectrum on the formal concept lattice of the data on the
      customers of a telecommunication company. Note that for the join spectrum, the lattice is inverted to have low frequencies on the bottom (see Table~\ref{summary}).}
    \label{fig:telcoMiddle}
\end{figure*}

We observe that the meet and join spectrum are fundamentally different but in both cases the signal is in tendency low frequency, which is more pronounced for the meet DLT. Further, the spectra reveal structure. They are Fourier-sparse and high values show transitions between property sets that are relevant for the churn variable. Building on DLSP one can imagine porting now basic SP methods for compression, sampling, or completion to social data analysis this way and also work on developing a deeper understanding on the meaning of spectrum beyond Section~\ref{sec:intuition}. Doing so is outside the scope of this paper, as we also wish to present another application domain in Section~\ref{sec:MultisetAuctions}.


\subsection{Relation to Bipartite Graphs and Hypergraphs} 

The relation ${\cal R}$ defines an undirected bipartite graph with nodes ${\cal O}\cup{\cal A}$ and an edge between two nodes if $(o,a)\in{\cal R}$. Conversely, every bipartite graphs defines a relation this way. ${\cal R}$ also defines a hypergraph\footnote{A hypergraph is a generalization of a graph that allows edges with more than two nodes.} with nodes ${\cal O}$ and hyperedges $\attr(a), a\in{\cal A}$ or nodes ${\cal A}$ and hyperedges $\obj(o), o\in{\cal O}$. Conversely, every hypergraph defines a relation this way. Thus, DLSP via FCLs, offers a form of Fourier analysis for these index domains, different from graph SP in \cite{Sandryhaila:13} or hypergraph SP in \cite{Zhang:20}.

\section{Application Example: Multiset Lattices in Spectrum Auctions}
\label{sec:MultisetAuctions}%

As a second prototypical application example we consider spectrum
auctions. We first show that bidders' preferences, called value
functions, in such auctions can be modeled as signals on multiset
lattices. Then we apply the sampling theorem to the problem of preference elicitation by reconstructing these value functions from few samples. Finally, we port Wiener
filtering to DLSP and use it to remove white noise from value functions.

\subsection{Multiset Lattices}
\label{subsec:MultisetLattice}%

We first introduce multiset lattices and then explain how signals on such lattices naturally appear as value functions in spectrum auctions.

\mypar{Lattice} A multiset is a generalization of a set that allows an element to appear more than once. Assuming a ground set $\{x_1,\dots,x_f\}$ of $f$ available elements, each (finite) multiset can be represented by a vector $m = (m_1,\dots,m_f)\in \mathbb{N}_0^f$, where each $m_i$ specifies how often element $i$ occurs in the set. For example, if $f=2$, $(2,1)$ represents $\{x_1,x_1,x_2\}$. A multiset $m$ is a set if all $m_i\leq 1$, i.e., if $m$ is a bit vector, a common representation for sets. 

A multiset $a$ is a submultiset of $m$ if $a \leq m$, defined componentwise. The set of submultisets $\mathcal{L} = \{a \in \mathbb{N}_0^f: a \leq m\}$ of a
given multiset $m$ is a lattice with meet $a \meet b = \min(a, b)$ and join
$a \join b = \max(a, b)$, for $a, b \in \mathcal{L}$, where the
minimum and maximum are taken componentwise. The multiset lattice
$\mathcal{L}$ has $|\mathcal{L}| = \prod_{1 \leq i \leq f} (m_i + 1)$
elements. A special case is the powerset lattice, for which $m_i = 1$
for all $i$ and thus its size is $2^f$. Note that if elements of a set are exchangeable, the multiset lattice has fewer elements than the powerset lattice.
Fig.~\ref{fig:PowersetVsMultiset} shows an example.

\begin{figure}\centering
    \hfill
    \subcaptionbox{Powerset lattice}[0.48\linewidth]
    {\includegraphics[width=\linewidth]{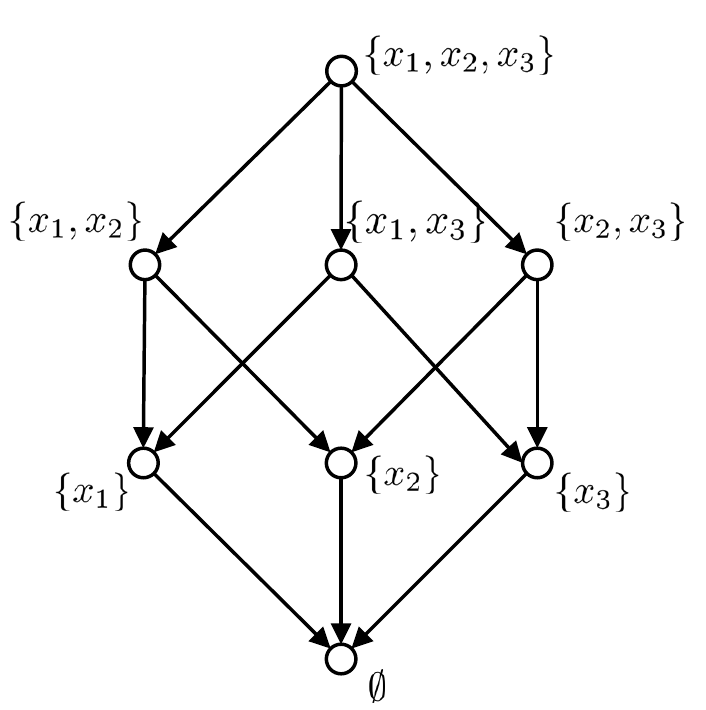}}
    \hfill
    \subcaptionbox{Multiset lattice}[0.48\linewidth]
    {\includegraphics[width=\linewidth]{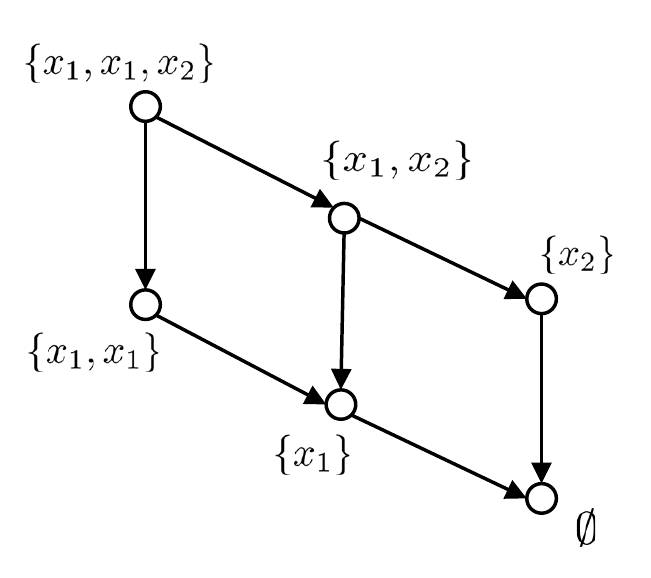}}
    \hfill
    \caption{If elements are exchangeable, e.g., if $x_1$ and $x_3$
      are the same thing, the powerset lattice~(a) requires more
      elements than the multiset lattice~(b) to model the situation
      accurately.\label{fig:PowersetVsMultiset}}
\end{figure}

\mypar{Spectrum auctions} Spectrum auctions~\cite{Cramton:13} are
combinatorial auctions~\cite{Ausubel:06} in which licenses for bands
of the electromagnetic spectrum are sold to bidders, e.g.,
telecommunication companies. For some bands there are multiple equal
licenses, which makes the set of available licenses a multiset $m$. 

\mypar{Lattice signal} Each bidder in such an auction is modeled as value function that assigns to each submultiset $a$ of available licenses its nonnegative value $s_a$ for the bidder. The value function is thus a multiset signal $\coord{s} = (s_a)_{a \in \latt}$, where $\latt$ is the lattice of submultisets of $m$ as explained before.

To find an allocation of the licenses to the bidders, an auctioneer first does a so-called preference elicitation. She asks each bidder for their values $s_a$ for a small number of $a\in\latt$, since the complete $\coord{s}$ is too large or too expensive to obtain. In other words, the auctioneer samples $\coord{s}$ for each bidder and an important problem is how to do this preference elicitation well \cite{Blum:04}.

As usual in the research on spectrum auctions, we consider simulated
bidders. In particular, here we use the single region value model (SRVM)
from the spectrum auctions test suite (SATS)~\cite{Weiss:17}. SATS
allows us to create multiple auction instances, and each instance comes with
its own set of bidders. There are four different bidder types: small, high-frequency, secondary, and primary bidders. The bidders' value functions depend on its type, e.g., small bidders prefer submultisets with few licenses, high-frequency bidders prefer licenses for high-frequency bands, etc. In SRVM the multiset of all available licenses is $m = (6, 14, 9)$. Modeling SRVM bidders' value functions as multiset functions instead of the commonly used set
functions~\cite{Weissteiner:20, Weissteiner:20a} thus reduces the
dimensionality of value functions from $2^{6+14+9} = 2^{29}$ to
$7 \cdot 15 \cdot 10 = 1050$.

In Fig.~\ref{fig:Bidder3SparseLattice} we show one example of a value function for one secondary bidder in SRVM and its (meet) Fourier transform. Relatively few of the frequencies contribute to the overall signal, i.e., the signal is Fourier-sparse as defined in Section~\ref{sec:sampling}. 

\begin{figure}
    \centering
    \hfill
    \includegraphics[width=0.49\linewidth]{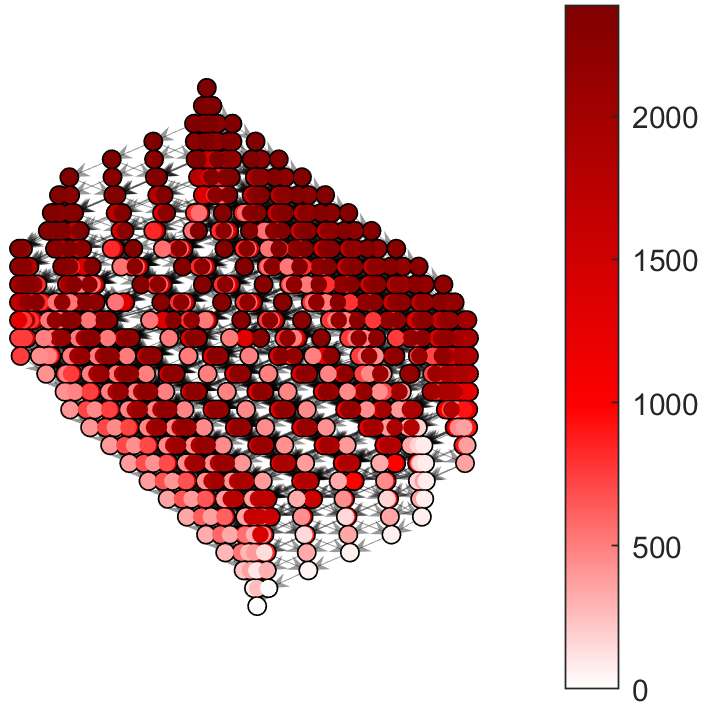}
    \hfill
    \includegraphics[width=0.49\linewidth]{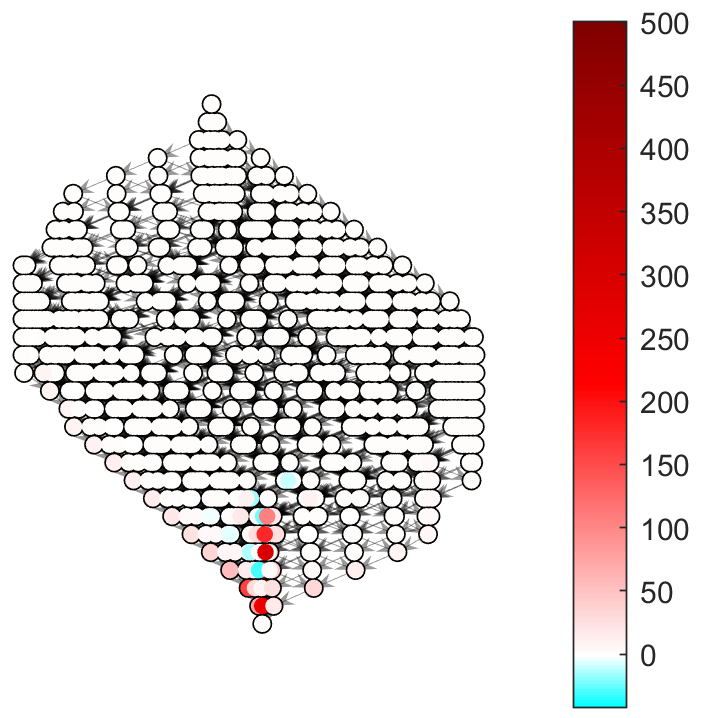}
    \hfill
    \caption{A secondary bidder as signal on a multiset lattice and its
      (meet) Fourier transform.}
    \label{fig:Bidder3SparseLattice}
\end{figure}

\subsection{Sampling in Combinatorial Auctions}
\label{subsec:SamplingAuctionsMultiset}%

In this section we apply our sampling Theorem~\ref{thm:sampling} to preference elicitation in combinatorial auctions.

\mypar{Experiment} Recently, several machine learning based preference
elicitation schemes~\cite{Blum:04}, where the bidders' preferences are modeled by,
e.g., neural networks~\cite{Weissteiner:20}, have been proposed. Here we sketch a
prototypical preference elicitation scheme based on DLSP Fourier sparsity.

Theorem~\ref{thm:sampling} can find a Fourier-sparse approximation
of a value function if the Fourier support of the most important
frequencies is known, which, in general is not the case. To overcome this issue, we compute the Fourier support from one auction instance of bidders and then use it for new auction instances and bidders, querying them according to Theorem~\ref{thm:sampling}.

It turns out that in the SRVM model all bidders were Fourier-sparse and their support did not change between different auction instances. Thus, we achieved perfect reconstruction from very few queries in all cases. We show our results in Table~\ref{tab:srvm}.

\begin{table}
    \ra{1.2}
    \centering
    \begin{tabular}{@{}lrrrr@{}}\toprule
      &  small & high-freq. & secondary & primary\\ \midrule
      Join-semilattice & 20 & 48 & 60 & 60 \\
      Meet-semilattice & 36 & 90 & 111 & 111 \\\bottomrule
    \end{tabular}
    \caption{Number of non-zero Fourier coefficients (= number of
      queries) required for the perfect reconstruction for different
      bidder types in SRVM.\label{tab:srvm}}
\end{table}

\subsection{Wiener Filter with Energy-Preserving Shift}
\label{sec:WienerFiltering}%

In this section we port Wiener filtering to DLSP using SRVM bidders as a case study. Our approach follows closely the one taken in graph SP by~\cite{Gavili.Zahng:2017a}. Namely, we first define an energy-preserving shift on which the Wiener filter is built. Part of this subsection is based on~\cite{Seifert.Wendler.Pueschel:2021a}.

\mypar{Energy-preserving shift} The lattice shifts, just as the graph shift, are not energy-preserving. Thus, similar to~\cite{Gavili.Zahng:2017a}
we first define an energy-preserving lattice shift as
\begin{equation}
    \label{eq:LatticeEnergyPreservingShift}
    T_e = \DLT^{-1} \cdot \Lambda_e \cdot \DLT
\end{equation}
with $\Lambda_e = \diag(\exp(-2 \pi \I k/\abs{\latt}) \; | \; k = 0,\dots,
\abs{\latt}-1)$. The frequency response of $T_e$ is the diagonal of $\Lambda_e$ and hence $T_e$ is a filter. The filter coefficients can be calculated
using~\eqref{invfreqresp}. Note that the complex frequency response also makes the filter coefficients in the lattice domain complex. Thus, we allow in this section the complex numbers as base field, to which the DLSP framework extends unmodified.

We summarize the properties of the energy-preserving shift.

\begin{theorem}[Properties of $T_e$]\label{theorem:EnergyPreservingShiftGeneratesLatticeAlgebra}%
The energy-preserving shift $T_e$ in \eqref{eq:LatticeEnergyPreservingShift}:
\begin{enumerate}[label=(\roman*)]
    \item \label{prop:EnergyPreserving}
    $\norm{\DLT \coord{s}}_2 = \norm{\DLT(T_e \coord{s})}_2$, i.e., it
    preserves energy, 
    \item \label{prop:GeneratesAllFilters} every lattice filter is a
    polynomial in the energy-preserving shift.
\end{enumerate}
\end{theorem}
\begin{proof}
    Property~\ref{prop:EnergyPreserving} follows from the simple
    derivation
    $$
    \begin{aligned}
        \norm{\DLT(T_e \coord{s})}_2 &= \norm{\DLT \DLT^{-1} \Lambda_e
          \DLT \coord{s}}_2 \\ 
        &= \norm{\Lambda_e \DLT \coord{s}}_2 = \norm{\DLT \coord{s}}_2, 
    \end{aligned}
    $$
    in which the last equality holds because $\Lambda_e$ is unitary.
    
    For~\ref{prop:GeneratesAllFilters} we need to find a polynomial
    $p_H$ to a given filter $H$ such that $p_H(T_e) = H$. Let $D_H = \DLT H \DLT^{-1}$ be the diagonal filter in the frequency domain. Thus, equivalently, we need to find $p_H$ with  $p_H(\Lambda_e) = D_H$. This is possible since all diagonal elements of $\Lambda_e$ are distinct and can be done, e.g., using Lagrange interpolation.
\end{proof}
Note that unlike the generating shifts, the energy-preserving shift
does not act sparsely, but is, in general, a full triangular matrix with
complex entries that acts globally (akin to the energy-preserving graph shift in \cite{Gavili.Zahng:2017a}). In other words $T_e$ is a non-sparse linear combination of all shifts in \eqref{shiftdefmat}:
$$
T_e = \sum_{x\in\latt}h_xT_x,\quad h_x\in\C.
$$

\mypar{Lattice Wiener filters} A Wiener filter is an optimal denoising
filter designed from a noisy version of a known reference signal. Such
a filter can handle noisy signals, which are similar to the reference
signal and where the noise is from the same noise model.

Consider a lattice signal $\coord{s}$ and a noisy measurement of the
signal $\coord{y} = \coord{s} + \coord{n}$. The lattice Wiener filter
of order $\ell$ based on the energy-preserving shift has then the form
\begin{equation}
    \label{eq:LatticeWienerFilter}
    H = \sum_{k=0}^{\ell} h_k T_e^k,
\end{equation}
where the filter coefficients $\coord{h}$ can be found by minimizing the
Euclidean error
\begin{equation}
    \label{eq:LatticeWienerFilterCoeffEquation}
    \min_{\coord{h}} \norm{H \coord{y} - \coord{s}}_2^2.
\end{equation}
We can rewrite~\eqref{eq:LatticeWienerFilterCoeffEquation} as
\begin{equation}
    \label{eq:LatticeWienerFilterCoefficients}
    \min_{\coord{h}} \norm{B \coord{h} - \coord{s}}_2^2,\quad\text{with }B =
    [\coord{y} \; T_e \coord{y} \ldots  \; T_e^{\ell} \coord{y}].
\end{equation}
The solution of this minimization problem is given by the solution of
the linear system
\begin{equation}
    \label{eq:WienerFilterCoefficientsEquation}
    B^H B \coord{h} = B^H \coord{s}
\end{equation}
for the coefficients $\coord{h}$ of the Wiener filter. Note that the
powers of $T_e$ can be computed efficiently in the frequency domain
using $T_e^k = \DLT^{-1} \Lambda_e^k \DLT$.

\mypar{Lattice white noise} White noise has equal intensity across
frequencies. In settings where the Fourier transform is orthogonal
(e.g., discrete time SP), white noise can be simulated by adding a
Gaussian noise vector with independent components to the signal. The
same procedure does not lead to white noise for a lattice signal, as
the $\DLT$ is not orthogonal. Thus we always add white noise directly
in the frequency domain. In Fig.~\ref{fig:ColoredNoise} we show an
example of white noise in lattice and frequency domain.
\begin{figure}
    \centering
    \hfill
    \subcaptionbox{\label{subfig:WhiteNoiseFrequency}}[0.49\linewidth]{
      \includegraphics[width=\linewidth]{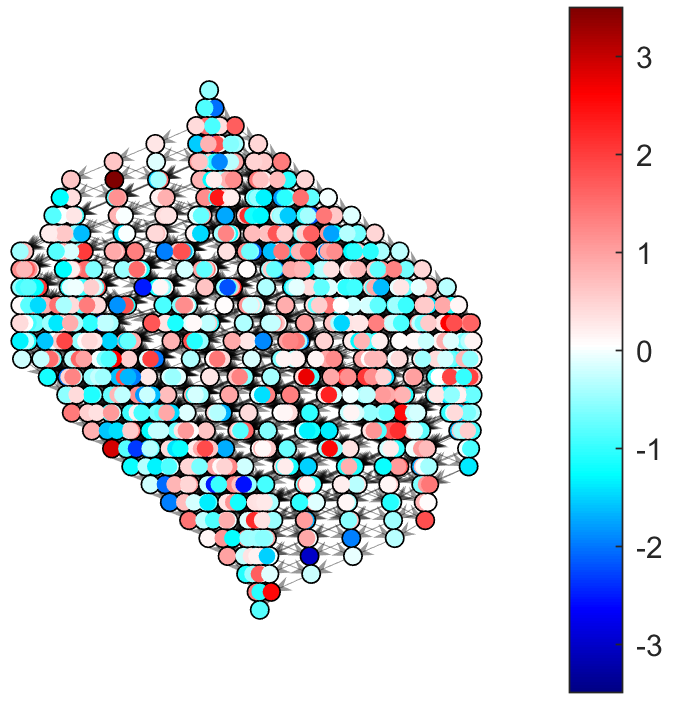} 
    }
    \hfill
    \subcaptionbox{\label{subfig:WhiteNoiseLattice}}[0.49\linewidth]{
      \includegraphics[width=\linewidth]{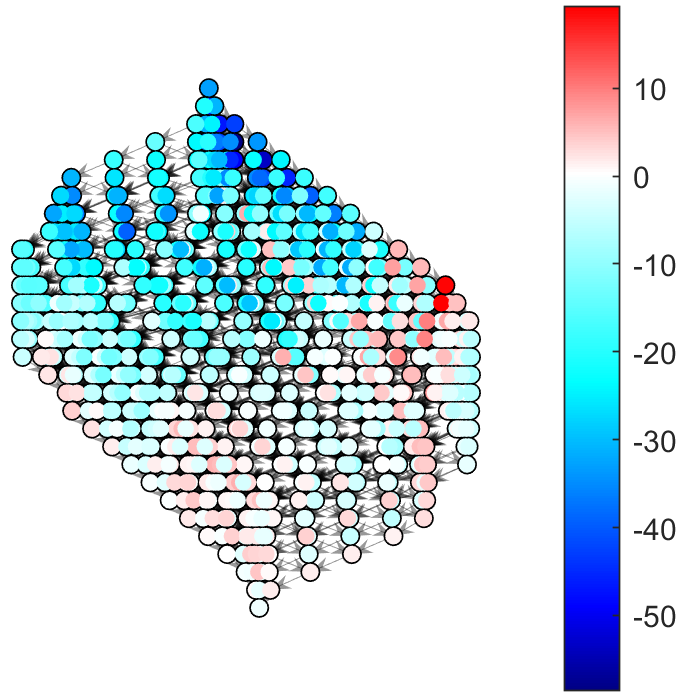}      
    }
    \hfill    
    \caption{White lattice noise in (a) frequency and (b) lattice
      domain.}
    \label{fig:ColoredNoise}
\end{figure}

\mypar{Experimental setup} For the experiment we used one secondary
bidder as reference signal $\coord{s}^{\text{ref}}$ and one primary
bidder as test signal $\coord{s}^{\text{test}}$. We added (different)
white lattice noise to the reference and test signal leading to a
signal-to-noise ratio of $12.5 \pm 2.5$ dB. We repeated the experiment
with 100 noise samples to obtain the standard deviation.

\mypar{Benchmark} As benchmark we compare the lattice Wiener filter to
a graph Wiener filter based on the cover graph as constructed in \cite{Gavili.Zahng:2017a}. Since the cover graph is directed and acyclic it is
not diagonalizable as required by~\cite{Gavili.Zahng:2017a}. Thus, we use the undirected cover graph instead. Note that the graph Wiener filter
has the disadvantage that lattice white noise cannot be modulated with
graph filters.

\mypar{Results} Fig.~\ref{fig:ResultsJoinAuctionsWienerFilter}
shows the relative reconstruction error as function of the filter
order, for the lattice and graph Wiener filtered signal. The
qualitative behavior on the reference signal for the graph Wiener
filter is as expected from~\cite{Gavili.Zahng:2017a}, but is
outperformed by the lattice Wiener filter. On the test signal the graph Wiener filter
raises the relative error, while the lattice Wiener filter reduces it.

\begin{figure}
    \centering
    \includegraphics[width=\linewidth]{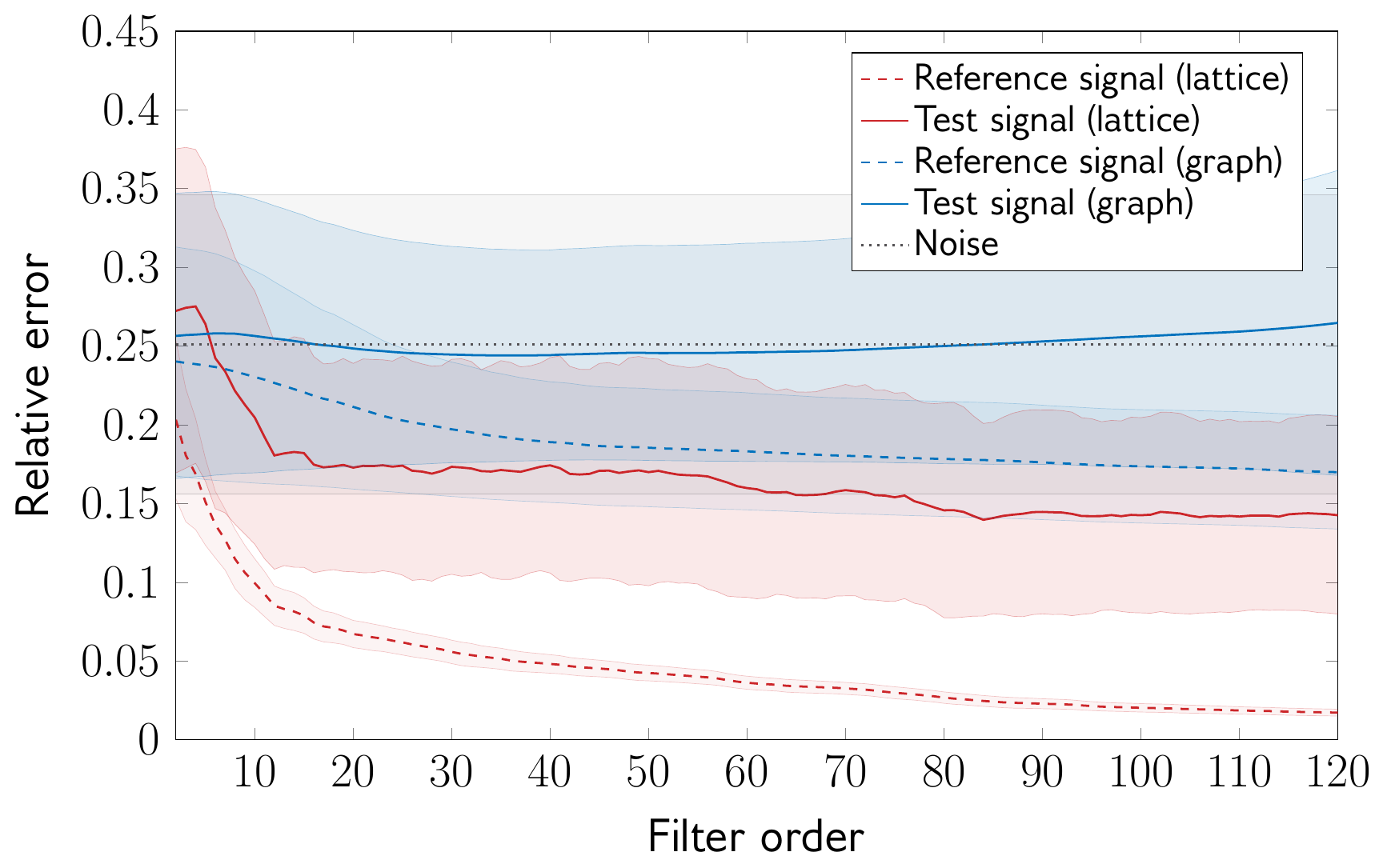}
    \caption{Denoising a bidder signal with lattice (red) and graph
      Wiener filter (blue) of different order. The results on the
      reference signals in both cases are shown dashed. The shaded
      areas are the standard deviations over 100 simulations.}
    \label{fig:ResultsJoinAuctionsWienerFilter}
\end{figure}

\section{Conclusion}

We presented discrete-lattice SP, a meaningful reinterpretation of fundamental SP concepts for signals indexed by partially ordered sets that support a meet (or join) operation. Lattices can be represented as cover graphs, but the lattice shifts operate very differently from graph shifts, capturing the partial order structure rather than proximity. Further, cover graphs are directed and acyclic and thus only have the eigenvalue zero, a problem for defining and computing a Fourier basis in graph SP.

Lattices are a fundamental structure in many domains as indicated in the introduction. In this paper we considered two examples: multiset lattices and formal concept lattices. We see particular potential in the latter, since with these lattices DLSP provides a general and novel form of Fourier analysis for data on relations, or, equivalently, bipartite graphs or hypergraphs.

Together with other novel SP frameworks that were recently introduced, DLSP shows how the power of traditional SP tools can be brought to new domains and to new kinds of data.

\ifCLASSOPTIONcaptionsoff
  \newpage
\fi



\bibliographystyle{IEEEtran}
\bibliography{paper}

\begin{IEEEbiography}[{\includegraphics[width=1in,height=1.25in,clip,keepaspectratio]{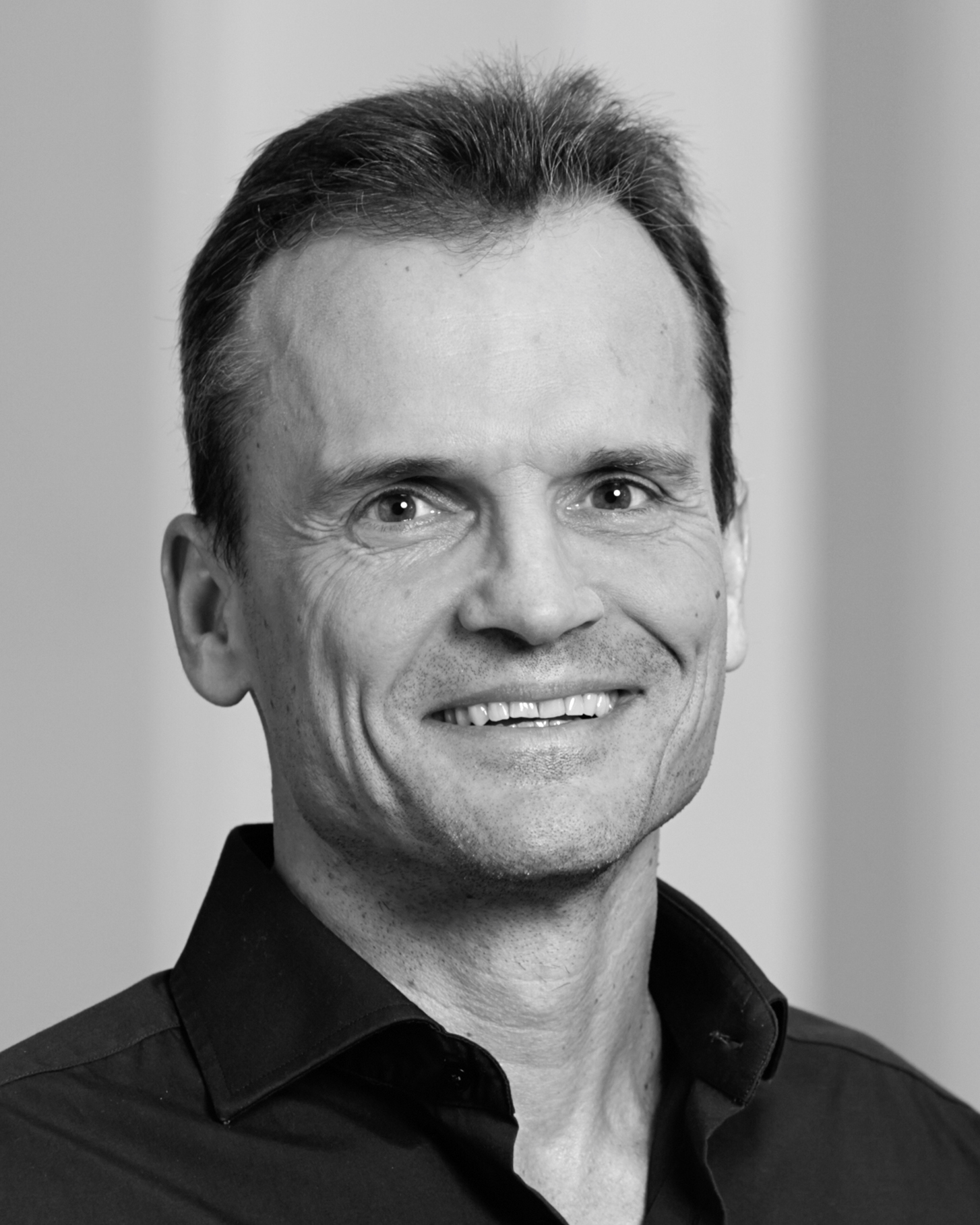}}]{Markus P\"uschel}
	(Fellow, IEEE) received the Diploma (M.Sc.) in mathematics and Doctorate
	(Ph.D.) in computer science, in 1995 and 1998, respectively, both from the University of Karlsruhe, Germany. He is a Professor of Computer Science with ETH Zurich, Switzerland, where he was the Head of the Department from 2013 to 2016. Before joining ETH in 2010, he was a Professor with Electrical and Computer Engineering, Carnegie Mellon University (CMU), where he still has an Adjunct status. He was an Associate Editor for the IEEE Transactions on Signal Processing, the IEEE Signal Processing Letters, and was a Guest Editor of the Proceedings of the IEEE and the Journal of Symbolic Computation, and served on numerous program committees of conferences in computing, compilers, and programming languages. He received the main teaching awards from student organizations of both institutions CMU and ETH and a number of awards for his research. His current research interests include algebraic signal processing, program generation, program analysis, fast computing, and machine learning.
\end{IEEEbiography}

\begin{IEEEbiography}[{\includegraphics[width=1in,height=1.25in,clip,keepaspectratio]{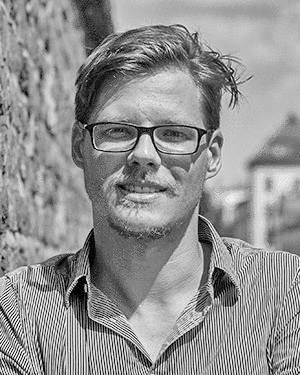}}]{Bastian  Seifert}
    (Member, IEEE) received the B.Sc. degree in mathematics from Friedrich-Alexander-Universität Erlangen-Nürnberg, Germany, in 2013, and the M.Sc. and Ph.D. degrees in mathematics from the Julius-Maximilians-Universität Würzburg, Germany, in 2015 and 2020, respectively. He was a Research Associate with the Center for Signal Analysis of Complex Systems (CCS), the University of Applied Sciences, Ansbach, Germany. Currently, he is a Postdoc at ETH Zurich, Switzerland. His research interests include algebraic signal processing, dimensionality reduction, and applied mathematics.
\end{IEEEbiography}

\begin{IEEEbiography}[{\includegraphics[width=1in,height=1.25in,clip,keepaspectratio]{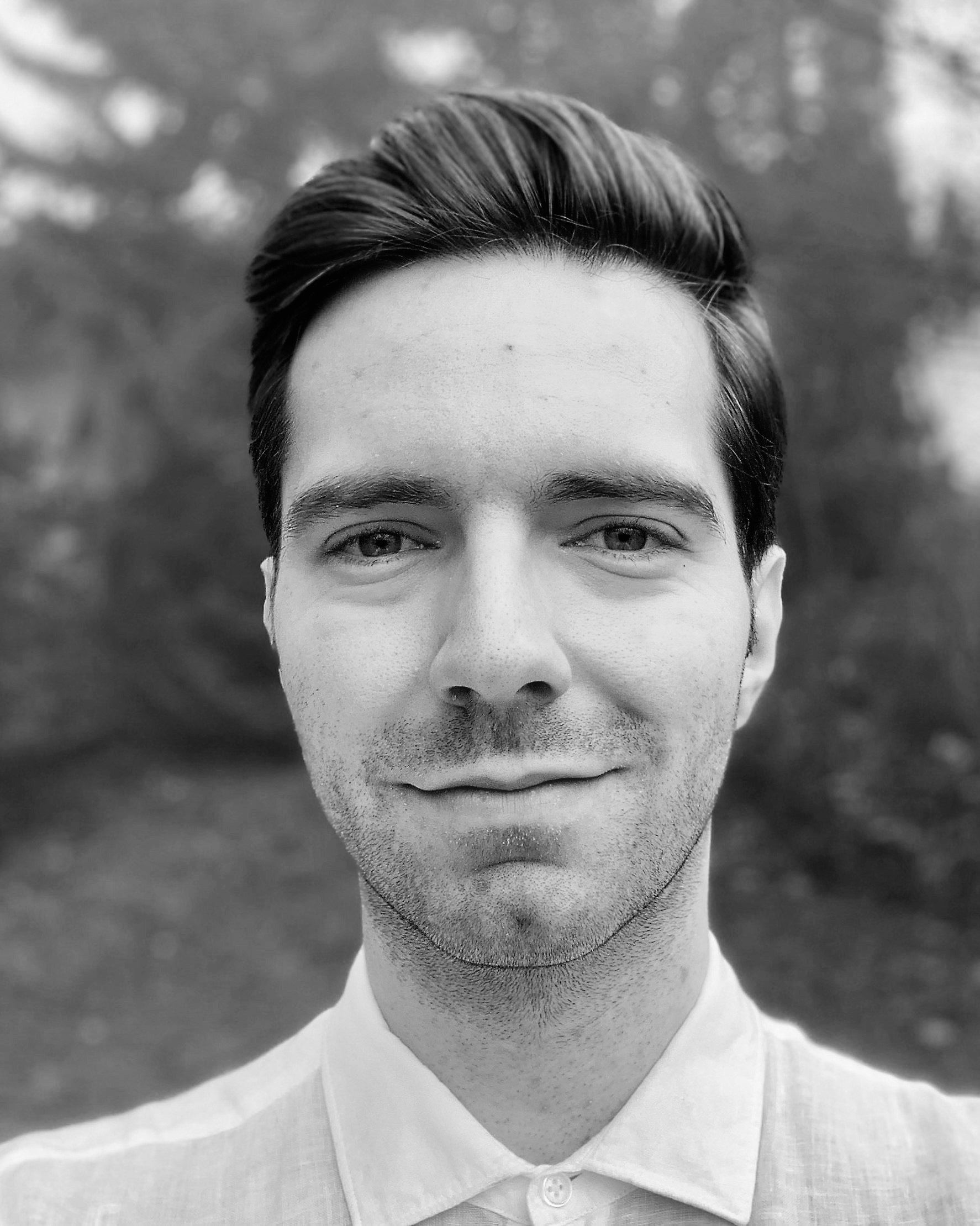}}]{Chris Wendler}
	received the B.Sc. and M.Sc. in computer science and the B.Sc. in mathematics, in 2013, 2016 and 2017, respectively, all from the Leopold-Franzens-University of Innsbruck, Austria. Currently, he is a Ph.D. candidate in the advanced computing laboratory at the department of computer science at ETH Zurich, Switzerland. His research focuses on signal processing and machine learning in non-Euclidean domains.
\end{IEEEbiography}
%
%
%

%



\vfill



%

\end{document}